\documentclass[12pt,a4paper,arial]{book}
\usepackage[utf8]{inputenc}
\usepackage{amsmath, amsfonts, amssymb, amsthm, framed}
\usepackage[margin=1in]{geometry}
\usepackage{amsmath}
\usepackage{breqn}
\usepackage{dsfont}
\usepackage{microtype}
\usepackage{xcolor}
\usepackage{courier} 
\usepackage{pgf,tikz,xcolor}
\usepackage{wrapfig}
\usetikzlibrary{decorations.text}
\usepackage{tikz} 
 
\usepackage{mathtools}

\newtheorem{theorem}{Theorem}
\usepackage{comment}
\usepackage{tikz}
\usetikzlibrary{shapes}
\usepackage{graphicx} 

\usepackage {tikz}
\usetikzlibrary {positioning}
\definecolor {processblue}{cmyk}{0.96,0,0,0}
\usepackage{hyperref} 
\usepackage{graphicx}
\usepackage{listings}
\usepackage{xcolor}
\usepackage{cleveref}
\usepackage{bm}
\usepackage{booktabs}          
\usepackage{array}             
\usepackage{caption}
\usepackage{amsthm}
\theoremstyle{plain}
\usepackage{xcolor}
\usepackage{tcolorbox}
\usepackage{authblk}

\title{\LARGE\bfseries An Introduction to Causal Modelling}

\author[1]{Gauranga Kumar Baishya%
  \thanks{\href{mailto:gauranga.mds2023@cmi.ac.in}{gauranga.mds2023@cmi.ac.in}}}
\author[1]{M.~R.~Srinivasan%
  \thanks{\href{mailto:mrsvasan@cmi.ac.in}{mrsvasan@cmi.ac.in}}}

\affil[1]{Chennai Mathematical Institute (CMI)}

\date{\today}

\begin{document}
\begin{titlepage}
  \centering
  \maketitle
  \thispagestyle{empty}
\end{titlepage}

\chapter*{\centering Acknowledgment}
This work is based on a ``Guided Project" where I read a couple of existing and latest literature in Causal Inference. I am profoundly grateful to Chennai Mathematical Institute and my advisor, Professor M.R.Srinivasan, for his sage advice, rigorous academic guidance, and the confidence he instilled in me. His expertise in the field of applied statistics and theory greatly enriched my work.

\begin{center}
    \vspace*{\fill}
    \tableofcontents
    \vspace*{\fill}
\end{center}
\newpage
\chapter{Introduction to Causal Inference}
\label{sec:intro}

\section*{Why Causality Matters}

Nearly every scientific or policy question that moves beyond description---from ``Will a carbon tax curb emissions?'' to ``Does aspirin prevent a second
heart attack?''--- is inherently causal.  Descriptive statistics can tell us what \emph{is}, but acting wisely requires knowing what \emph{would happen} under  different actions. Causal inference thus seeks to answer counterfactual questions:  
\begin{quote}
\centering
\itshape
What would the outcome have been, had a different choice been made?
\end{quote}

Historically, the formal study of causality lay scattered across philosophy, epidemiology, statistics, and economics.  Modern frameworks---notably the
potential‑outcomes model and graphical DAG approaches---have unified these threads, giving us a precise language, clear assumptions, and powerful tools for estimation and criticism. The remainder of this single high‑level chapter develops the essential building blocks of causal thinking, weaving together intuitive stories, mathematical notation, and concrete examples.

\section{Potential Outcomes and Individual Causal Effects}
\label{subsec:po}

\paragraph{Historical roots.}
The potential‑outcomes idea dates to Neyman(1923) in agricultural field trials and was popularized by Rubin in the 1970s.  Its genius is to treat each unit as carrying a \emph{vector} of latent responses---one for every treatment value---thereby separating the science (how Nature assigns those latent
responses) from the statistics (how the analyst observes and aggregates them).


\paragraph{Notation and basic definitions.}
Let \(A \in \{0,1\}\) be a \emph{binary treatment} (with \(A=1\) for the active treatment and \(A=0\) for control) and let  
\(Y \in \{0,1\}\) be a \emph{binary outcome} measured at the end of the study.  
For every individual \(i\) we postulate two \emph{potential outcomes}
\[
  Y_i^{1}, \quad Y_i^{0} 
  \qquad 
  (\text{equivalently written } Y_i(a)\text{ for } a \in \{0,1\}),
\]
and collect them in the vector of potential outcomes
\[
  \mathbf{Y}_i \;=\; \bigl(Y_i^{0},\,Y_i^{1}\bigr).
\]

\begin{enumerate}
  \item These potential outcomes are \textbf{well‑defined} only under the  
        \emph{Stable Unit Treatment Value Assumption (SUTVA)}:
        \begin{enumerate}
            \item \emph{No multiple versions of treatment}: each value of \(A\) corresponds to a unique, well‑specified intervention.
            \item \emph{No interference}: an individual’s outcome depends only on their own treatment assignment, not on the assignments of others.
        \end{enumerate}
  \item The observed outcome is a \emph{reveal} of exactly one coordinate:
        \[
          Y_i \;=\; Y_i^{A_i}.
        \]
        The unobserved coordinate is called the \emph{counterfactual} outcome.
\end{enumerate}

\medskip
\noindent
The \textbf{individual causal effect} is defined as the contrast
\[
  \Delta_i 
  \;=\; 
  Y_i^{1} \;-\; Y_i^{0} 
  \;\in\; \{-1,\,0,\,1\},
\]
indicating whether, for that person, treatment
\begin{enumerate}
  \item harms (\(\Delta_i = -1\)),
  \item has no effect (\(\Delta_i = 0\)), or
  \item benefits (\(\Delta_i = 1\)).
\end{enumerate}

\smallskip
\noindent
Because, for each individual, we can ever observe \emph{at most one} of the two potential outcomes, the pair \(\bigl(Y_i^{0},Y_i^{1}\bigr)\) is never jointly observed, making \(\Delta_i\) \emph{intrinsically unidentifiable}.  
This is not merely a statistical limitation but a logical one, often referred to as the \emph{fundamental problem of causal inference}.

\paragraph{Example.}
Zeus undergoes a heart transplant (\(A = 1\)) and dies on the fifth
post‑operative day (\(Y = 1\)).  Consequently we observe the treated
potential outcome
\[
  Y_{\text{Zeus}}^{1} \;=\; 1,
\]
while the untreated potential outcome \(Y_{\text{Zeus}}^{0}\) is \emph{fundamentally unobservable}.  Any statement about what would have happened to Zeus \emph{without} transplantation must rely on data from \emph{other} patients who did not receive a transplant, together with
assumptions (e.g.\ conditional ignorability) that make Zeus ``exchangeable'' with those patients.  In this way, information is \emph{borrowed} from comparable untreated individuals to draw causal conclusions about Zeus’s counterfactual outcome.

\section{Average Causal Effects and Effect‐Measure Scales}
\label{subsec:ace}

\paragraph{From individuals to populations.}
Recall that the \emph{individual causal effect} 
\(\Delta_i = Y_i^1 - Y_i^0\) 
is never fully observed for any one person.  Nevertheless, scientific inference and policy decisions typically target \emph{average} effects over a population or subpopulation.  The most fundamental estimand is the 
\emph{average treatment effect} (ATE), defined by
\[
\mathrm{ATE}
\;=\;
\mathbb{E}\bigl[Y^1 - Y^0\bigr]
\;=\;
\mathbb{E}[Y^1] \;-\; \mathbb{E}[Y^0].
\]
In many applications one also considers two refined, subpopulation‐specific estimands:
\begin{description}
  \item[\(\mathrm{ATT}\)]%
    \emph{Average Treatment Effect on the Treated:}
    \[
      \mathrm{ATT}
      \;=\;
      \mathbb{E}\bigl[Y^1 - Y^0 \mid A = 1\bigr],
    \]
    which answers, “What is the average effect among those who actually received treatment?”
    
  \item[\(\mathrm{ATC}\)]%
    \emph{Average Treatment Effect on the Untreated (Controls):}
    \[
      \mathrm{ATC}
      \;=\;
      \mathbb{E}\bigl[Y^1 - Y^0 \mid A = 0\bigr],
    \]
    useful for planning expansions of a program to currently untreated individuals.
\end{description}

\paragraph{Risk‐based summary measures.}
When the outcome \(Y\) is binary, additional \emph{scale choices} help convey effect direction and public‐health impact.  Three common measures are:

\begin{description}
  \item[Risk Difference (RD):]%
    \[
      \mathrm{RD}
      \;=\;
      P(Y^1 = 1) \;-\; P(Y^0 = 1).
    \]
    This \emph{absolute} scale aligns with notions such as the “number needed to treat” 
    \(\bigl(\mathrm{NNT} = 1/\mathrm{RD}\bigr)\).

  \item[Risk Ratio (RR):]%
    \[
      \mathrm{RR}
      \;=\;
      \frac{P(Y^1 = 1)}{P(Y^0 = 1)}.
    \]
    A value below 1 (e.g.\ \(\mathrm{RR}=0.6\)) can be read as “40\% relative reduction.”

  \item[Odds Ratio (OR):]%
    \[
      \mathrm{OR}
      \;=\;
      \frac{P(Y^1 = 1)/P(Y^1 = 0)}{P(Y^0 = 1)/P(Y^0 = 0)}.
    \]
    Widely used in case–control studies and logistic regression.  Note that for \emph{common} outcomes (prevalence \(>10\%\)), the OR can substantially overstate the RR.
\end{description}

\paragraph{Choosing a scale.}
The statistical efficiency, interpretability, and relevance of an effect‐measure scale depend on the scientific context (additive vs.\ multiplicative mechanism), the target audience, and the study design.  Always match the scale to both the underlying causal model and the stakeholders’ needs. 

\section{Association versus Causation}
\label{subsec:assoc}

\paragraph{Conditional versus Marginal Risk.}

Consider two distinct notions of risk:

\[
\underbrace{P(Y = 1 \mid A = a)}_{\text{observed conditional risk}}\quad \text{versus}\quad 
\underbrace{P(Y^{a} = 1)}_{\text{counterfactual (marginal) risk}}.
\]

The first quantity, \( P(Y = 1 \mid A = a) \), is an \emph{observed conditional risk}. It represents the probability of the outcome given that treatment status \(A\) is observed to be \(a\). Thus, this measure is derived solely from individuals who \emph{actually received} treatment \(a\).

In contrast, the second quantity, \( P(Y^{a} = 1) \), is a \emph{counterfactual (marginal) risk}, representing the probability that the outcome would be observed if the entire population were universally assigned to treatment \(a\). This counterfactual risk cannot typically be observed directly, since it involves imagining the outcome under a hypothetical intervention, regardless of individuals' actual treatment assignments.

These two risks are conceptually distinct. The equality:

\[
P(Y = 1 \mid A = a) \;=\; P(Y^{a} = 1)
\]

holds only under strong conditions, notably \emph{unconfoundedness}, also known as \emph{conditional exchangeability}, where treatment assignment is independent of potential outcomes. Such conditions are naturally satisfied in randomized trials but are rarely guaranteed in observational studies, making direct equivalence of these risks uncommon in practice. Careful methods of adjustment, conditioning, or weighting are typically required to make valid causal inferences from observational data.

\paragraph{Confounding and Simpson's Paradox.}
A variable \( L \) is said to \emph{confound} the relationship between treatment \( A \) and outcome \( Y \) if it simultaneously influences both \( A \) and \( Y \) and is not an intermediate variable on the causal pathway from \( A \) to \( Y \). Confounding thus generates a non-causal association, distorting our ability to discern the true causal effect of treatment.

\href{https://en.wikipedia.org/wiki/Simpson%27s_paradox}{Simpson's paradox} vividly illustrates the perils of unadjusted analyses: it occurs when the direction of the association between treatment and outcome observed in the overall population reverses once we stratify (or condition) on the confounder \( L \). This paradox underscores the crucial need for proper adjustment or stratification in observational studies. To systematically identify potential confounding and clarify causal structures, researchers frequently rely on Directed Acyclic Graphs (DAGs). These graphical tools offer intuitive and rigorous methods for pinpointing pathways that can introduce bias, ensuring that analyses adjust for the right variables and preserve the integrity of causal inference.

\paragraph{Aspirin:}
Let’s take the example of aspirin treatment (\( A \)) and the chance of dying (\( Y \)) among heart patients. Patients who are already at high risk are more likely to be given aspirin. Because these patients are already more likely to die, we might see a higher death rate among those who got aspirin: \( P(Y = 1 \mid A = 1) > P(Y = 1 \mid A = 0) \). At first glance, this could wrongly make it seem like aspirin increases the risk of death. But that’s not true—it’s just that sicker patients were more likely to get the drug. If we adjust for how sick the patients were at the start (using something like a "risk score"), or if we randomly assign who gets aspirin and who doesn’t, we find the true story: aspirin actually helps. After accounting for this confounding, we see that aspirin lowers the risk of death—it has a protective, causal effect.

\bigskip
Thus, Raw associations in data are often misleading. Without appropriate analytical methods or experimental designs that explicitly control for confounding, we risk drawing incorrect conclusions. Robust causal inference requires thoughtfully severing all spurious pathways and carefully distinguishing genuine causal effects from mere correlations.

\section{RCT, Exchangeability, and Consistency}
\label{subsec:rct}

\paragraph{Design Logic.}
Randomized controlled trials (RCTs) are the gold standard for causal inference because randomization severs all backdoor paths from baseline prognostic variables to treatment. This guarantees marginal exchangeability:
\[
Y^{a} \perp\!\!\!\perp A,
\]
meaning that the potential outcomes are independent of treatment assignment. This condition, combined with two other key assumptions—\textbf{consistency}, 
\[
Y = Y^A,
\]
which links the observed outcome to the potential outcome corresponding to the treatment actually received, and \textbf{positivity},
\[
0 < P(A = a \mid L) < 1 \quad \text{for all strata } L,
\]
forms the standard identification triad. These assumptions allow us to estimate causal effects from observed data in an RCT setting.

\paragraph{Intention-to-Treat (ITT) vs Per-Protocol:}
In practice, some participants may not follow their assigned treatment (non-compliance). The \textbf{Intention-to-Treat (ITT)} principle analyzes participants according to their original treatment assignment, regardless of whether they adhered. This approach preserves the exchangeability provided by randomization, avoiding bias due to post-randomization behavior. In contrast, a \textbf{Per-Protocol} analysis compares participants based on the treatment actually received. While this may reflect the \emph{causal efficacy} of the treatment under ideal use, it requires strong assumptions to adjust for any post-randomization confounders of adherence. Without proper adjustment, the causal interpretation of per-protocol estimates may be invalid.

\paragraph{Blinding and Concealment:} \textbf{Blinding} protects the study from differential behavior or measurement that could arise if participants or researchers know the treatment assignment. For example, caregivers might unconsciously deliver co-interventions differently, or outcome assessors might introduce bias.
\textbf{Allocation concealment} ensures that treatment assignment is not known at the point of enrollment, preventing selection bias. 

\bigskip

Thus, randomization is not just about flipping a fair coin—it involves a bundle of practices designed to maintain the integrity of causal claims.

\section{Conditional Randomization, Stratification, and Blocking}
\label{subsec:block}

\paragraph{} In many randomized controlled trials, treatment is not assigned completely at random across the entire study population. Instead, subjects are often grouped into strata (or "blocks") based on observed covariates such as age, sex, or disease severity. Within each block, treatment is then randomized. This is known as a \textbf{blocked design} or \textbf{stratified randomization}.Blocking improves the precision of causal estimates by ensuring that treatment groups are balanced with respect to key baseline characteristics. In such designs, treatment assignment is conditionally independent of the potential outcomes given the strata, i.e.,
\[
Y^{a} \perp\!\!\!\perp A \mid L,
\]
where \(L\) is the blocking (or stratifying) variable. However, marginal (unconditional) independence \(Y^{a} \perp A\) may no longer hold. Because treatment assignment is randomized \emph{within strata} but not across the entire population, the analysis must account for the blocking. Ignoring this design structure can result in biased or inefficient estimates. Standard approaches include stratified analyses such as the Mantel–Haenszel estimator or regression models with interaction terms for each stratum.

\paragraph{Heart-Transplant Example in Depth.}
Consider a heart transplant trial where \(A = 1\) indicates receiving a transplant and \(Y = 1\) indicates death. Let \(L = 1\) represent critically ill patients and \(L = 0\) represent stable patients. In this trial, suppose:
\[
P(A = 1 \mid L = 1) = 0.75 \quad \text{and} \quad P(A = 1 \mid L = 0) = 0.50.
\]
This indicates that sicker patients are more likely to be assigned to the treatment group, making the overall (crude) comparison misleading. For instance, the crude mortality rates are:
\[
\text{Treated: } 7/13 \quad \text{vs} \quad \text{Untreated: } 3/7,
\]
suggesting a higher risk of death in the treated group. However, when we stratify by \(L\), the mortality rates within each group are:
\[
\text{For } L = 1: \quad 2/3 \text{ in both treated and control,}
\]
\[
\text{For } L = 0: \quad 1/4 \text{ in both treated and control.}
\]

Within each stratum, the treatment effect disappears—mortality risk is identical. After standardization (i.e., averaging the stratum-specific risks according to the distribution of \(L\)), we find that the marginal (overall) risk ratio is actually 1. This illustrates how conditioning on \(L\) can reveal the true causal story hidden beneath the crude associations.

\paragraph{Standard Error Implications.}
Ignoring the blocking structure in analysis can lead to:
\begin{enumerate}
    \item Inflated standard error estimates, as the precision gains from stratification are lost;
    \item Missed detection of effect modification, where the treatment effect might vary across strata;
    \item Potential confounding if the stratification variable is related to both treatment and outcome.

To avoid these pitfalls, it is critical that statistical models reflect the study design. At a minimum, the analysis should include indicator variables (dummy variables) for each block, and, if warranted, interaction terms between treatment and blocks to test for effect heterogeneity. Proper model specification ensures valid inference and optimizes the use of stratified designs.
\end{enumerate}

\section{Crossover Experiments and Period Effects}
\label{subsec:xover}

\paragraph{Statistical model.}
Let $t\in\{0,1\}$ index periods, $A_{it}$ the treatment at period $t$, and
$Y_{it}$ the response.  Under assumptions of
\emph{no carryover},
\emph{equal period effects}, and
\emph{time‑stable causal effect} $\alpha_i$,
the within‑subject difference
\(
D_i=Y_{i1}-Y_{i0}
\)
identifies $\alpha_i$ free of between‑person confounders.

\paragraph{Practical wrinkles.}
Wash‑out periods mitigate biological carryover; period effects (fatigue,
learning) call for including a fixed effect for time.  Missingness, often
non‑random, erodes the self‑controlled virtue and demands inverse‑probability
weighting in turn.

\paragraph{When infeasible.}
Irreversible treatments (organ transplant) or terminal outcomes (death)
preclude crossover schemes; alternative self‑controlled designs (case‑crossover,
N‑of‑1) may still apply when outcomes are transient.

\section{Inverse Probability Weighting (IPW) in Practice}
\label{subsec:ipw2}

\paragraph{Estimand via Weighting.}
Inverse Probability Weighting (IPW) is a widely used method for estimating average treatment effects in observational studies. The key idea is to create a pseudo-population in which the distribution of observed covariates is independent of treatment assignment—mimicking a randomized trial.

Let \(\pi(L) = P(A = 1 \mid L)\) denote the \textbf{propensity score}, i.e., the probability of receiving treatment given baseline covariates \(L\). The propensity score captures the treatment selection mechanism. We then define the \textbf{stabilized IPW weights} for individual \(i\) as:
\[
w_i = \frac{A_i \cdot P(A = 1)}{\pi(L_i)} + \frac{(1 - A_i) \cdot P(A = 0)}{1 - \pi(L_i)}.
\]
These weights aim to balance the covariate distributions across treatment groups. Stabilized weights are preferred over naive weights like \(1/\pi(L_i)\), as they help reduce variance inflation and yield more stable estimators in finite samples.

\paragraph{Positivity and Truncation.}
The success of IPW depends crucially on the \textbf{positivity assumption}, which requires that:
\[
0 < \pi(L_i) < 1 \quad \text{for all } i.
\]
That is, every individual must have a non-zero probability of receiving either treatment or control, regardless of their covariates \(L\). In practice, some values of \(\pi(L_i)\) may be very close to 0 or 1, leading to extremely large weights and inflated standard errors. To address this, analysts often perform \textbf{weight truncation} or \textbf{weight stabilization}, for example, by capping weights at the 1st and 99th percentiles. This introduces a small amount of bias but substantially reduces the variance and leads to more reliable inference. A more serious problem arises with \textbf{structural positivity violations}, where certain strata of \(L\) have no treated or no control observations. In such cases, it is fundamentally impossible to estimate treatment effects in those regions of the covariate space. The solution is to \emph{redefine} the causal question and restrict attention to the population where treatment comparisons are possible—``where the data support the effect.'' We will see more in Chapter 4.

\paragraph{Sandwich Variance and Bootstrap.}
Since the propensity score \(\pi(L)\) is typically estimated from data (e.g., via logistic regression or machine learning), the resulting weights \(w_i\) are random variables and introduce additional uncertainty.
To account for this, standard variance formulas must be adjusted. Two common methods for robust inference are:

\begin{enumerate}
    \item \textbf{Sandwich variance estimators} (also known as robust or heteroskedasticity-consistent estimators), which correct the standard errors of IPW estimators by accounting for the estimated nature of the weights.
    
    \item \textbf{Non-parametric bootstrap}, which resamples the data to empirically estimate the distribution of the IPW estimator. This approach is flexible and widely used, especially when machine learning methods are used to estimate \(\hat\pi(L)\).

\end{enumerate}

\section{References}
\begin{enumerate}
 \item James M.\ Robins and Miguel A.\ Hernán (2020).  
        \emph{Chapter 1, 2, 3, What If? The Foundations of Causal Inference}. Chapman and Hall/CRC.  
        Also freely available at: \url{https://miguelhernan.org/whatifbook}
\end{enumerate}
\newpage

\newpage
\chapter{Randomized Experiments}

\label{subsec:smoking}

Let \( D \in \{0,1\} \) denote a binary treatment variable indicating whether an individual smokes marijuana. Specifically, \( D = 1 \) denotes a smoker and \( D = 0 \) denotes a non-smoker. Following the potential outcomes framework, we associate with each individual two hypothetical (counterfactual) outcomes:

\[
Y(1) \quad \text{and} \quad Y(0),
\]where \( Y(1) \) is the individual's lifespan if they smoke, and \( Y(0) \) is the lifespan if they do not. However, in practice, we only observe one of these outcomes per individual, depending on their actual smoking status. To formalize this, we define a structural model for potential outcomes:

\begin{equation}
Y(d) = \eta_0 + \eta_1 d
\label{eq:struct}
\end{equation} Here, \( \eta_0 \) is a person-specific intercept capturing their baseline health endowment (e.g., genetics, early childhood environment), and \( \eta_1 \) reflects the biological effect of marijuana smoking on longevity. For simplicity and didactic purposes, we assume \( \eta_1 = 0 \), implying that marijuana smoking has no causal effect on lifespan. In other words, the only differences in observed lifespans will arise due to differences in \( \eta_0 \), not due to the treatment itself.

\bigskip

Now suppose that smoking behavior is not randomly assigned but arises from an individual's latent propensity or predisposition to smoke. This is modeled as:

\begin{equation}
D = \mathbf{1}\{ \nu > 0 \},
\label{eq:selection}
\end{equation} where \( \nu \) represents unobserved psychological factors such as risk tolerance, peer influence, or personality traits. Crucially, we assume that this unobserved propensity is \emph{negatively} correlated with baseline health:

\begin{equation}
E[\eta_0 \nu] < 0.
\label{eq:negcor}
\end{equation} This means individuals who are more likely to smoke tend to start off with worse health on average (due to for example, drinking). This assumption introduces selection bias, as the treatment assignment (\( D \)) is no longer independent of potential outcomes. Under our assumption that \( \eta_1 = 0 \), the potential outcomes simplify to:

\[
Y(1) = Y(0) = \eta_0,
\] and hence the observed outcome is simply:

\[
Y = Y(D) = \eta_0.
\] However, in observational data, we do not observe \( \eta_0 \) or \( \nu \) directly. The only observed data are the realized pairs \( (Y, D) \), where \( D \) is the observed treatment (smoking status), and \( Y \) is the observed lifespan.

\paragraph{Apparent Health Gap.}
Even though smoking has no causal effect, selection into smoking based on unobserved health characteristics distorts the observed comparison of average lifespans between smokers and non-smokers. We can compute the expected outcome (lifespan) for each group:

\begin{align}
E[Y \mid D = 1] 
  &= E[\eta_0 \mid \nu > 0] 
  \;<\; E[\eta_0], \label{eq:EY1} \\
E[Y \mid D = 0] 
  &= E[\eta_0 \mid \nu \leq 0] 
  \;>\; E[\eta_0]. \label{eq:EY0}
\end{align} This leads to the observed inequality:

\begin{equation}
E[Y \mid D = 1] < E[Y \mid D = 0], \tag*{(*)}
\end{equation} suggesting that smokers have shorter lifespans than non-smokers. But this is an illusion. Since the treatment \( D \) was not randomized and is negatively correlated with baseline health (\cref{eq:negcor}), the difference is entirely due to \emph{selection bias}, not any biological harm from marijuana use.

\bigskip

This toy example illustrates the fundamental problem of causal inference from observational data. The \emph{observed} difference in means:

\[
\pi = E[Y \mid D = 1] - E[Y \mid D = 0],
\] which we call the \textbf{Average Predictive Effect} (APE), does not equal the \textbf{Average Treatment Effect} (ATE):

\[
\delta = E[Y(1)] - E[Y(0)],
\] whenever the treatment \( D \) is \emph{endogenously} selected. That is, when treatment choice is influenced by variables that also affect the outcome, simple comparisons of outcomes between groups will confound correlation with causation. Without randomization or adequate adjustment for confounding, we cannot identify causal effects reliably.

\section{Selection Bias: Conceptual Foundations}

Formally, bias arises whenever
\[
E[Y(d)\mid D=1] \;\neq\;E[Y(d)],
\qquad d\in\{0,1\}.
\]
The inequality reflects the non‑independence of \(D\) and \(Y(d)\).  Let
\(\Delta_{\text{bias}}\) denote
\[
\Delta_{\text{bias}}=E[Y\mid D=1]-E[Y(1)],
\]
so that APE \(=\delta+\Delta_{\text{bias}}\).  Strategies to neutralise
\(\Delta_{\text{bias}}\) include:

\begin{enumerate}
  \item \textbf{Design:} randomized controlled trials (RCTs), encouragement designs, cluster randomization.
  \item \textbf{Re‑weighting:} inverse‑probability weights, propensity‑score matching, doubly robust estimators.
  \item \textbf{Regression adjustment:}—conditional exchangeability within covariate strata.
\end{enumerate} Among these, RCTs are the cleanest because they modify the data‑generating
process itself rather than adjusting after the fact.

\section{RCT: Assumptions and Theorem}
\label{subsec:rct-assumptions}

\paragraph{Fundamental Assumption.}
In a well-executed randomized controlled trial (RCT), the treatment assignment \( D \) is independent of all potential outcomes. That is, the potential outcomes that would occur under each treatment level do not influence the probability of receiving the treatment. Formally, this is expressed as:

\[
\boxed{D \;\perp\!\!\!\perp\; Y(1),\, Y(0)}.
\] Additionally, the \textbf{positivity assumption} ensures that both treatment arms are represented in the data:

\[
0 < P(D = 1) < 1.
\]

\paragraph{Randomization Theorem.}
Under the independence assumption, the observed mean outcomes in each treatment group can be interpreted as the mean potential outcomes. Specifically, for \( d \in \{0,1\} \),

\[
E[Y \mid D = d] 
= E[Y(d) \mid D = d] 
= E[Y(d)],
\] which implies that:

\[
\pi = E[Y \mid D = 1] - E[Y \mid D = 0] = \delta,
\] where \( \pi \) is the observed average difference between treated and control groups, and \( \delta \) is the \textbf{average treatment effect} (ATE). Thus, randomization ensures that selection bias is eliminated by design.

\paragraph{Note.}
The theoretical independence guaranteed by randomization holds only when several practical conditions are met:

\begin{enumerate}
    \item \textbf{Allocation concealment:} Treatment assignment must be hidden from participants and investigators to prevent manipulation or selection bias.
    \item \textbf{No differential attrition:} Dropout rates should not differ systematically between treatment arms, as this can reintroduce bias.
    \item \textbf{Intention-to-treat (ITT) analysis:} Participants should be analyzed based on their original treatment assignment, regardless of compliance. This preserves the benefits of randomization.
\end{enumerate}

\section{Sampling Framework for an RCT}

We observe
\[
\bigl\{(Y_i,D_i)\bigr\}_{i=1}^{n},
\]
i.i.d.\ from the super‑population distribution of \((Y,D)\).  Throughout we
treat \(n_1=\sum D_i\) and \(n_0=n-n_1\) as random binomial counts.

\section{Point Estimation of Group Means}
\label{subsec:group-means}

In the context of treatment effect estimation, one of the most fundamental tasks is to estimate the expected outcome for each treatment group. For a binary treatment variable \( D \in \{0, 1\} \), we define the group-specific population means as:

\[
\theta_d = \mathbb{E}[Y \mid D = d], \qquad \text{for } d = 0, 1.
\] These quantities represent the average outcomes we would observe within each treatment arm if the entire population were observed under treatment status \( D = d \).

\paragraph{Sample Estimators.}
Given data from a sample of \( n \) individuals indexed by \( i = 1, \ldots, n \), a natural and unbiased estimator of \( \theta_d \) is the empirical mean outcome within the subgroup that received treatment status \( d \). This is given by:

\[
\hat{\theta}_d
  = \frac{\sum_{i=1}^{n} Y_i \cdot \mathbf{1}\{D_i = d\}}
         {\sum_{i=1}^{n} \mathbf{1}\{D_i = d\}}
  = \frac{\bar{Y}_d}{\hat{p}_d},
\]

where:
\begin{enumerate}
  \item \( \mathbf{1}\{D_i = d\} \) is the indicator function that equals 1 when individual \( i \) received treatment \( d \), and 0 otherwise.
  \item \( \hat{p}_d = \frac{n_d}{n} \) is the empirical proportion of individuals in the sample assigned to treatment arm \( d \), with \( n_d = \sum_{i=1}^n \mathbf{1}\{D_i = d\} \).
  \item \( \bar{Y}_d = \frac{1}{n_d} \sum_{i: D_i = d} Y_i \) is the sample mean outcome among individuals who received treatment \( d \).
\end{enumerate} These estimators \( \hat{\theta}_0 \) and \( \hat{\theta}_1 \) are simple group means and serve as foundational building blocks for estimating causal effects, such as the average treatment effect (ATE), which can be computed as \( \hat{\theta}_1 - \hat{\theta}_0 \). They are unbiased estimators of the conditional expectations \( \mathbb{E}[Y \mid D = d] \), provided that sampling is random and no informative missingness or selection bias is present.

\section{Inference for the Average Treatment Effect}
\label{subsec:ate-inference}

The \textbf{average treatment effect (ATE)} is defined as the difference in group means:
\[
\delta = \theta_1 - \theta_0,
\] where \( \theta_d = \mathbb{E}[Y \mid D = d] \) for \( d = 0,1 \). The natural sample estimator is:
\[
\hat{\delta} = \hat{\theta}_1 - \hat{\theta}_0,
\] where \( \hat{\theta}_d \) is the sample mean of \( Y \) in treatment group \( d \), as discussed previously. Assuming finite variances within each treatment arm, the multivariate Central Limit Theorem (CLT) applies:
\[
\sqrt{n}
\begin{bmatrix}
\hat{\theta}_0 - \theta_0 \\
\hat{\theta}_1 - \theta_1
\end{bmatrix}
\xrightarrow{d}
\mathcal{N}\!\left(
  \begin{bmatrix}
  0 \\ 0
  \end{bmatrix},
  \begin{bmatrix}
    \sigma_0^2 / p_0 & 0 \\
    0 & \sigma_1^2 / p_1
  \end{bmatrix}
\right),
\]
where:
\begin{enumerate}
  \item \( \sigma_d^2 = \operatorname{Var}(Y \mid D = d) \) is the conditional variance,
  \item \( p_d = \mathbb{P}(D = d) \) is the marginal probability of treatment arm \( d \).
\end{enumerate} Applying the delta method (specifically for linear combinations), we derive the asymptotic distribution of \( \hat{\delta} \):
\[
\sqrt{n}(\hat{\delta} - \delta)
  \xrightarrow{d}
  \mathcal{N}\left( 0, \; \frac{\sigma_0^2}{p_0} + \frac{\sigma_1^2}{p_1} \right).
\]

\paragraph{Wald Confidence Interval.}
To construct confidence intervals for \( \delta \), we estimate the variances \( \sigma_d^2 \) using their sample analogues:
\[
\hat{\sigma}_d^2 = \frac{1}{n_d - 1} \sum_{i: D_i = d} (Y_i - \bar{Y}_d)^2.
\] A large-sample (asymptotic) \( 100(1 - \alpha)\% \) Wald confidence interval for \( \delta \) is given by:
\[
\hat{\delta} \pm z_{\alpha/2}
\sqrt{
  \frac{\hat{\sigma}_0^2}{n_0}
  + \frac{\hat{\sigma}_1^2}{n_1}
},
\] where \( z_{\alpha/2} \) is the standard normal quantile corresponding to confidence level \( 1 - \alpha \), and \( n_d \) is the number of observations in group \( d \).

\bigskip

This framework provides valid inference under random sampling and large-sample approximations. In small samples or complex designs, bootstrap or robust methods may be preferred.

\section{The Delta Method and Relative Effectiveness}
\label{subsec:delta-method}

In many applied contexts, especially in health economics, epidemiology, and policy evaluation, we are interested not just in the absolute effect of a treatment, but in its \emph{relative} effect. This is often used to express efficiency or cost-effectiveness. Define the relative treatment effect as the proportionate change in the mean outcome:
\[
\varphi = \frac{\theta_1 - \theta_0}{\theta_0},
\] where \( \theta_1 = \mathbb{E}[Y \mid D=1] \) and \( \theta_0 = \mathbb{E}[Y \mid D=0] \) are the population means for the treated and control groups, respectively. Let us define a function \( f : \mathbb{R}^2 \to \mathbb{R} \) that maps the parameter vector \( \bm\theta = (\theta_0, \theta_1)^\top \) to the relative effect:
\[
f(\bm\theta) = \frac{\theta_1 - \theta_0}{\theta_0}.
\]
This simplifies to:
\[
f(\bm\theta) = \frac{\theta_1}{\theta_0} - 1.
\]

\paragraph{Delta Method Application.}
To obtain the large-sample distribution of the estimator \( f(\hat{\bm\theta}) \), we use the delta method. First, compute the gradient vector:
\[
G = \nabla f(\bm\theta)
= \begin{bmatrix}
    \frac{\partial f}{\partial \theta_0} \\
    \frac{\partial f}{\partial \theta_1}
  \end{bmatrix}
= \begin{bmatrix}
    -\frac{\theta_1}{\theta_0^2} \\
    \frac{1}{\theta_0}
  \end{bmatrix}.
\] The asymptotic distribution of \( \hat{\bm\theta} = (\hat{\theta}_0, \hat{\theta}_1)^\top \),
\[
\sqrt{n}(\hat{\bm\theta} - \bm\theta)
\xrightarrow{d}
\mathcal{N}\left( \mathbf{0}, V \right),
\]
where \( V \) is the asymptotic covariance matrix of \( \hat{\bm\theta} \), given by:
\[
V =
\begin{bmatrix}
    \sigma_0^2 / p_0 & 0 \\
    0 & \sigma_1^2 / p_1
\end{bmatrix}.
\] Then by the delta method, the relative effect estimator satisfies:
\[
\sqrt{n}\left( f(\hat{\bm\theta}) - f(\bm\theta) \right)
\xrightarrow{d}
\mathcal{N}(0,\, G^\top V G).
\]

\paragraph{Use Cases:}
This approach allows us to construct standard errors and confidence intervals for multiplicative effect measures such as:
\begin{enumerate}
    \item \textbf{Benefit-Cost Ratios} in economics,
    \item \textbf{Vaccine Efficacy}, typically computed as \(1 - \text{RR}\),
    \item \textbf{Rate Ratios or Relative Risks} in epidemiological studies.
\end{enumerate} It is particularly useful when the absolute effect is less interpretable or when proportional impacts provide more policy-relevant insights.

\section{Covariate Adjustment and Treatment Effect}
\label{subsec:covariate-adjustment}

In many randomized experiments and observational studies, baseline covariates \( W \) — such as age, gender, income, health status, or education — may influence both the outcome and the treatment effect. To account for this, we define the \textbf{Conditional Average Treatment Effect (CATE)}:
\[
\delta(W) = \mathbb{E}[Y(1) - Y(0) \mid W].
\]
This quantity captures how the causal effect of treatment varies across individuals or subgroups characterized by different values of \( W \). In essence, it generalizes the average treatment effect (ATE) to a conditional level, allowing for \emph{treatment effect heterogeneity}.

Under the assumption of random treatment assignment (e.g., in a randomized controlled trial), the treatment indicator \( D \in \{0,1\} \) is independent of potential outcomes conditional on \( W \). That is,
\[
Y(1), Y(0) \perp\!\!\!\perp D \mid W.
\]
This implies the equality:
\[
\mathbb{E}[Y \mid D = d, W] = \mathbb{E}[Y(d) \mid W],
\quad \text{for } d \in \{0,1\}.
\]
Consequently, we can define the \textbf{Conditional Average Predictive Effect}:
\[
\pi(W) = \mathbb{E}[Y \mid D = 1, W] - \mathbb{E}[Y \mid D = 0, W],
\]
and under ignorability, we obtain:
\[
\pi(W) = \delta(W).
\]
That is, the difference in conditional means from observed data correctly identifies the causal effect for covariate strata.

\paragraph{Efficiency gains through covariate adjustment.}
Even in randomized trials, adjusting for baseline covariates can lead to efficiency gains. By modeling outcome variation across covariates using methods such as ANCOVA (Analysis of Covariance), regression adjustment, or modern machine learning methods like Targeted Minimum Loss Estimation (TMLE) or double machine learning, one can:
\begin{enumerate}
    \item Reduce residual variance,
    \item Improve precision of treatment effect estimates,
    \item Construct narrower confidence intervals,
    \item Detect heterogeneity across subgroups.
\end{enumerate}
These techniques leverage the predictive power of covariates \( W \) to "de-noise" the treatment effect estimation.

If covariates \( W \) contain missing values, simply removing such observations (complete-case analysis) may introduce selection bias — especially if the missingness mechanism is \textbf{not completely at random}, for instance, if it depends on both the treatment assignment \( D \) and the outcome \( Y \). In such settings, valid inference requires more sophisticated strategies:
\begin{enumerate}
    \item \textbf{Multiple Imputation:} Fill in missing values using posterior draws from predictive models.
    \item \textbf{Inverse Probability Weighting for Missingness:} Weight complete cases by the inverse of their response probabilities.
    \item \textbf{Doubly Robust Methods:} Combine imputation and weighting to ensure consistency even if one of the models is misspecified.
\end{enumerate}

\bigskip
Covariate adjustment is a powerful tool for improving efficiency and exploring heterogeneity in treatment effects. However, proper handling of missing covariate data is essential to maintain the validity of causal conclusions.

\section{Limitations of RCTs}

\paragraph{Violations of SUTVA.}
The Stable Unit Treatment Value Assumption has two clauses:
(i) no interference among units,
(ii) no hidden versions of treatment.  
Herd‑immunity spill‑overs, market‑equilibrium effects, and varying
implementation fidelity each breach SUTVA, complicating causal interpretation.

\paragraph{Ethical constraints.}
Randomly assigning harmful exposures (smoking initiation, poverty) is
unacceptable; equipoise, informed consent, and beneficence govern study design.
When RCTs are impossible, quasi‑experimental designs—difference‑in‑differences,
instrumental variables, regression discontinuity—step in.

\paragraph{Cost and power.}
Field trials can be logistically daunting and expensive, especially for rare
outcomes requiring huge \(n\) for adequate power.  Cluster randomization eases
logistics but inflates variance via intraclass correlation.

\paragraph{External validity.}
An RCT estimates the ATE for its sampling frame, yet effects may attenuate or
amplify when scaled up.  Transportability analysis confronts such ‘‘all else
equal’’ fallacies.  
\section{Conclusion \& References}
Randomization underpins modern drug licensing (FDA), evidence‑based development economics (Banerjee, Duflo, Kremer), and behavioural interventions (Thaler’s nudge agenda).  While critics highlight cost and context specificity, the clarity of the causal identification it affords keeps RCTs the gold standard against which observational methods are benchmarked.
\begin{enumerate}
  \item Victor Chernozhukov et al.\ (2023).  
        \emph{Applied Causal Inference Powered by Machine Learning and AI},
        Chapter 2.  
        \url{https://causalml-book.org/}
  \item Sylvain Chabé‑Ferret (2024).  
        \emph{Statistical Tools for Causal Inference}, Chapter 3: Randomized
        Controlled Trials.  
        \url{https://chabefer.github.io/STCI/RCT.html}
  \item Imbens \& Rubin (2015).  
        \emph{Causal Inference for Statistics, Social, and Biomedical
        Sciences}.  Cambridge UP.
\end{enumerate}
\newpage

\chapter[Double ML]{Double Lasso, Neyman Orthogonality, and Related Methods}

\label{subsec:highdim-breakdown}

In classical linear regression theory, the number of covariates \( p \) is assumed to be small relative to the sample size \( n \). Under this assumption, the Ordinary Least Squares (OLS) estimator is well-behaved: it is unbiased, efficient, and asymptotically normal under standard regularity conditions. This framework underlies much of traditional econometric analysis. However, a few modern data applications often violate the low-dimensional assumption. For example:
\begin{enumerate}
    \item In genetics, DNA microarray data may involve \( p \sim 10^4 \) gene expression levels for a few hundred patients.
    \item In digital marketing, clickstream data can generate \( p \sim 10^5 \) binary indicators from user behavior.
    \item In cross-country development research, datasets may include dozens of institutional indicators across fewer than only 200 countries.
\end{enumerate} When the number of covariates \( p \) is large relative to, or even exceeds, the sample size \( n \), traditional OLS estimation fails in several ways:
\begin{enumerate}
    \item The matrix \( (X^\top X)^{-1} \) becomes ill-conditioned or undefined when \( p \ge n \), making the OLS solution either unstable or non-existent.
    \item Including all variables indiscriminately leads to inflated variance and overfitting, degrading out-of-sample predictive performance.
    \item Conversely, performing aggressive pre-screening of covariates risks omitting important variables, introducing bias in coefficient estimates.
\end{enumerate} To address these challenges, modern regression analysis often relies on \textbf{regularization techniques}, which penalize model complexity and encourage sparsity. Among the most widely used is the \textbf{Least Absolute Shrinkage and Selection Operator (Lasso)}. Lasso estimation solves the optimization problem:
\[
\hat{\beta}^{\text{lasso}} = \arg\min_{\beta} \left\{ \frac{1}{2n} \| Y - X\beta \|_2^2 + \lambda \|\beta\|_1 \right\},
\]
where \( \lambda \) is a tuning parameter controlling the degree of shrinkage. The Lasso estimator offers key advantages in high-dimensional settings:
\begin{enumerate}
    \item It performs automatic variable selection by shrinking some coefficients exactly to zero.
    \item It provides sparse solutions that improve interpretability and predictive performance.
\end{enumerate} However, Lasso introduces a critical trade-off. The \emph{shrinkage bias} it induces pulls estimated coefficients toward zero, which invalidates standard inferential tools such as \( t \)-tests and confidence intervals derived from OLS theory. This creates a fundamental tension: how can we retain the predictive power and sparsity benefits of Lasso while also conducting valid statistical inference on key parameters—such as the effect of a treatment or policy variable—amid high-dimensional confounding? Answering this question leads us toward \textbf{double selection methods}, \textbf{debiased Lasso}, and \textbf{double machine learning}, which build on Lasso's strengths while correcting for its limitations.
\section{Frisch–Waugh–Lovell (FWL) Partialling‑Out}
\label{subsec:fwl-partialling}

Consider the linear regression model:
\begin{equation}
Y = \alpha D + \beta^\top W + \varepsilon,
\label{eq:model}
\end{equation}
where:
\begin{enumerate}
    \item \( Y \) is the outcome variable of interest (e.g., economic growth, wage, blood pressure, etc.),
    \item \( D \in \mathbb{R} \) is the \textbf{target regressor} whose coefficient \( \alpha \) represents the causal or predictive effect we wish to estimate,
    \item \( W \in \mathbb{R}^p \) is a high-dimensional vector of covariates or controls,
    \item \( \varepsilon \) is an idiosyncratic disturbance satisfying the orthogonality condition \( \mathbb{E}[\varepsilon \mid D, W] = 0 \).
\end{enumerate}

\paragraph{Frisch–Waugh–Lovell (FWL) Theorem:} In classical low-dimensional settings, \(\hat{\alpha}_{\text{OLS}} \) can be equivalently obtained by the following two-step residual regression procedure:
\begin{enumerate}
    \item First, regress \( Y \) on \( W \) and obtain the residuals \( \tilde{Y} = Y - \hat{\pi}_Y^\top W \).
    \item Then, regress \( D \) on \( W \) and compute the residuals \( \tilde{D} = D - \hat{\pi}_D^\top W \).
    \item Finally, regress \( \tilde{Y} \) on \( \tilde{D} \) to obtain \( \hat{\alpha} \) from the simple regression:
    \[
    \tilde{Y} = \alpha \tilde{D} + \varepsilon.
    \]
\end{enumerate} By removing the linear effects of \( W \) from both the outcome \( Y \) and the regressor of interest \( D \), we isolate the component of \( D \) that is uncorrelated with the confounders \( W \), allowing us to estimate its effect on \( Y \) without bias from omitted variables. When the number of controls \( p \) becomes large relative to the sample size \( n \), classical OLS techniques break down. The matrix \( (W^\top W)^{-1} \) becomes ill-conditioned or non-invertible, and even if inversion is numerically feasible, the variance of \( \hat{\alpha} \) can explode due to overfitting. As a result, standard OLS-based residualization becomes unreliable or impossible.

\bigskip
To overcome these limitations, modern econometrics uses regularized regression—specifically, the \textbf{Double Lasso} or \textbf{Partialling-Out Lasso} method. Instead of projecting \( Y \) and \( D \) onto \( W \) via OLS, we perform two separate Lasso regressions.

\bigskip
First, we run a Lasso of $Y$ on $W$:
\[
\hat{\gamma}_{Y}
  \;=\;
  \arg\min_{\gamma\in\mathbb{R}^{p}}
  \Bigl\{\sum_{i=1}^{n}(Y_i-\gamma^{\!\top}W_i)^2
        \;+\;\lambda_1\sum_{j=1}^{p}\psi_{Yj}\,|\gamma_j|
  \Bigr\}.
\] The penalty loadings $\psi_{Yj}$ (often $\psi_{Yj}=1/\hat\sigma_j$) scale the
$\ell_1$ constraint so that units of different regressors are comparable. Independently we fit a Lasso of $D$ on $W$:
\[
\hat{\gamma}_D
  \;=\;\arg\min_{\gamma}\,
  \Bigl\{
    \sum_{i}(D_i-\gamma^{\!\top}W_i)^2
    +\lambda_2\sum_{j}\psi_{Dj}\,|\gamma_j|
  \Bigr\}.
\]
Both~$\lambda_1$ and~$\lambda_2$ can be chosen via modified cross‑validation or
the theoretically motivated formula
$\lambda\asymp\sigma\sqrt{2\log p/n}$. We then compute
\[
\check Y_i = Y_i - \hat{\gamma}_Y^{\!\top} W_i,
\quad
\check D_i = D_i - \hat{\gamma}_D^{\!\top} W_i.
\]
These are the pieces of $Y$ and $D$ orthogonal, \emph{up to Lasso error}, to
the entire $W$ space. Regressing the residualized outcome on residualized treatment gives
\[
\hat\alpha
  \;=\;
  \frac{\sum_i \check D_i\check Y_i}{\sum_i\check D_i^{\,2}}
  \;=\;
  \bigl(\check D^{\!\top}\check D\bigr)^{-1}
  \bigl(\check D^{\!\top}\check Y\bigr).
\] Because $(\check D_i)$ is a scalar, the matrix inverse above is merely a division.

\paragraph{Intuition.}
The union of variables selected in Step 1 or Step 2 acts like \emph{double selection}. Any regressor important for predicting either $Y$ \emph{or} $D$ is implicitly adjusted for.  Consequently the omitted variable bias of missing ``weak but critical'' controls is dramatically
reduced compared with ``single'' Lasso selection that screens only on $Y$.

\section{Neyman Orthogonality.}
The success of Double Lasso relies on a powerful statistical property called \textbf{Neyman orthogonality}. It ensures that small errors in the estimation of nuisance components (i.e., the Lasso fits for \( Y \) and \( D \), discussed below) do not significantly bias the final estimate of \( \alpha \). Specifically, the estimation procedure is \emph{first-order insensitive} to small perturbations in the nuisance parameters—an essential requirement for valid inference in high-dimensional models. In contrast, methods that do not satisfy Neyman orthogonality—such as naive single Lasso selection—allow first-stage bias in $\hat{\eta}$ to leak directly into $\hat{\alpha}$ through linear terms. This contamination prevents the estimator from achieving $\sqrt{n}$-rate convergence and undermines valid statistical inference. Thus, Neyman orthogonality is not merely a technical requirement but a foundational design principle that enables double machine learning procedures to perform reliable inference in high-dimensional settings where traditional OLS methods fail.
We prove it next. Let $\alpha(\eta)$ be the target parameter defined implicitly by the moment condition
\[
M(\alpha,\eta) := \mathbb{E}\left[(\tilde{Y}(\eta) - \alpha \tilde{D}(\eta)) \tilde{D}(\eta)\right] = 0,
\]
where
\[
\tilde{Y}(\eta) = Y - \eta_1^\top W, \qquad \tilde{D}(\eta) = D - \eta_2^\top W.
\]
Here, $\eta = (\eta_1, \eta_2)$ represents a collection of nuisance parameters, and $\eta^o$ denotes their true values. The property of \textbf{Neyman orthogonality} requires that the estimator $\alpha(\eta)$ be locally insensitive to small perturbations in $\eta$, i.e.,
\[
\partial_{\eta} \alpha(\eta^o) = 0.
\]

In high-dimensional settings, nuisance parameters $\eta$ are typically estimated (e.g., via Lasso), introducing slight bias. Neyman orthogonality ensures that these small estimation errors do not influence $\alpha$ to first order. Formally, for a small perturbation $\delta$ around the true nuisance value $\eta^o$,
\[
\alpha(\eta^o + \delta) = \alpha(\eta^o) + \underbrace{\partial_\eta \alpha(\eta^o) \cdot \delta}_{=0} + \text{higher-order terms}.
\]

\begin{proof}

Since $\alpha(\eta)$ is implicitly defined via the moment equation $M(\alpha,\eta) = 0$ and the \textbf{implicit function theorem} gives:
\[
\partial_{\eta} \alpha(\eta^o) = -\left[\partial_{\alpha} M(\alpha,\eta^o)\right]^{-1} \partial_{\eta} M(\alpha,\eta^o).
\]
To show that $\partial_{\eta} \alpha(\eta^o) = 0$, it suffices to demonstrate:
\[
\partial_{\eta} M(\alpha,\eta^o) = 0.
\] Recall that
\[
M(\alpha,\eta) = \mathbb{E}\left[(\tilde{Y}(\eta) - \alpha \tilde{D}(\eta)) \tilde{D}(\eta)\right],
\]
with
\[
\tilde{Y}(\eta) = Y - \eta_1^\top W, \qquad \tilde{D}(\eta) = D - \eta_2^\top W.
\] Differentiating with respect to $\eta_1$ we get, 
\[
\partial_{\eta_1} \tilde{Y}(\eta) = -W \quad \Rightarrow \quad \partial_{\eta_1} M(\alpha,\eta^o) = \mathbb{E}\left[-W \cdot \tilde{D}(\eta^o)\right].
\]
Since $\tilde{D}(\eta^o)$ is the residual from regressing $D$ on $W$, it is orthogonal to $W$:
\[
\mathbb{E}\left[W \cdot (D - \gamma_{D,W}^\top W)\right] = 0,
\]
so
\[
\partial_{\eta_1} M(\alpha,\eta^o) = 0.
\] and the Derivative with respect to \(\eta_2\) gives, 
\[
\partial_{\eta_2} \tilde{D}(\eta) = -W.
\]So, the product rule gives
\begin{align*}
\partial_{\eta_2} M(\alpha,\eta^o)
&= \mathbb{E} \left[ \partial_{\eta_2} \left( \tilde{Y}(\eta^o) \cdot \tilde{D}(\eta^o) \right) - \alpha \cdot \partial_{\eta_2} \left( \tilde{D}(\eta^o)^2 \right) \right] \\
&= - \mathbb{E} \left[ W \cdot \tilde{Y}(\eta^o) \right] + 2\alpha \cdot \mathbb{E} \left[ W \cdot \tilde{D}(\eta^o) \right].
\end{align*}
Again, $\tilde{Y}(\eta^o)$ and $\tilde{D}(\eta^o)$ are residuals from projecting $Y$ and $D$ on $W$, respectively. Therefore:
\[
\mathbb{E}[W \cdot \tilde{Y}(\eta^o)] = 0, \quad \mathbb{E}[W \cdot \tilde{D}(\eta^o)] = 0.
\]
So,
\[
\partial_{\eta_2} M(\alpha,\eta^o) = 0.
\] Since both partial derivatives vanish, and it follows from the implicit function theorem that:
\[
\partial_{\eta} \alpha(\eta^o) = - \left[\partial_{\alpha} M(\alpha,\eta^o)\right]^{-1} \cdot 0 = 0.
\]
    
\end{proof}


\section{Asymptotic Theory: \texorpdfstring{$\sqrt{n}$}{sqrt-n}‑Normality of \texorpdfstring{$\hat\alpha$}{alpha}}
\label{subsec:asymptotic_normality}

Under appropriate sparsity and moment assumptions—such as the condition that the sparsity level $s$ satisfies ($s \log p \ll n$—Belloni, Chernozhukov, and Hansen (2014)) a central limit theorem for the double machine learning estimator $\hat\alpha$ can be established.

\bigskip
Let $\tilde D$ denote the residual from the population-level projection of $D$ on $W$, and let $\varepsilon$ denote the structural error from the partially linear model
\[
Y = \alpha D + \beta^\top W + \varepsilon, \qquad \mathbb{E}[\varepsilon \mid D, W] = 0.
\]
Then the debiased or orthogonalized estimator $\hat\alpha$ satisfies:
\[
\sqrt{n}(\hat\alpha - \alpha)
  \xrightarrow{d}
  \mathcal{N}\left(0,\;
    V := \mathbb{E}[\tilde D^2]^{-1} \,
         \mathbb{E}[\tilde D^2 \varepsilon^2] \,
         \mathbb{E}[\tilde D^2]^{-1}
  \right).
\] A consistent plug-in estimator $\hat{V}$ replaces the population expectations with their sample analogues:
\[
\hat{V} = \left(\frac{1}{n} \sum_{i=1}^n \tilde D_i^2 \right)^{-2}
\left( \frac{1}{n} \sum_{i=1}^n \tilde D_i^2 \hat\varepsilon_i^2 \right),
\]
where $\hat\varepsilon_i = Y_i - \hat\alpha \cdot D_i - \hat\beta^\top W_i$ is the estimated residual from the final stage. Using this variance estimator, a Wald-type $95\%$ confidence interval for $\alpha$ is given by:
\[
\boxed{
  \hat\alpha \;\pm\; 1.96\,\sqrt{\frac{\hat V}{n}}
}
\]
This interval has asymptotically correct coverage, with error rate converging to zero at the standard $\mathcal{O}(n^{-1/2})$ rate.

\bigskip
A key strength of this double machine learning procedure is its robustness to imperfect model selection. The validity of inference for $\alpha$ does not require perfect variable selection in the high-dimensional regressions of $Y$ and $D$ on $W$. Instead, it suffices that the Lasso-based nuisance estimators achieve a sufficiently accurate approximation of the best $s$-sparse linear predictors. This flexibility is due to the Neyman orthogonality property, which guarantees that small estimation errors in the first stage do not translate into first-order biases in the second-stage estimate of the target parameter.

\section{Conditional Convergence in Growth Economics}
\label{subsec:conditional_convergence}

A central hypothesis in development economics is the theory of \textbf{conditional convergence}, which posits that poorer countries tend to grow faster than richer ones, \emph{conditional} on similar structural characteristics. In other words, once we control for institutional quality, educational attainment, trade openness, and other growth determinants, we expect countries with lower initial income to exhibit higher subsequent growth rates. This hypothesis translates into a regression framework where the dependent variable is the economic growth rate and the main regressor of interest is initial GDP per capita:
\[
\text{Growth}_i
  = \alpha\,\text{InitialGDP}_i
    + \beta^{\!\top}W_i + \varepsilon_i,
\]
where:
\begin{enumerate}
  \item $\text{Growth}_i$ is the GDP growth rate for country $i$ over a fixed period (e.g., 10 years),
  \item $\text{InitialGDP}_i$ is the log of initial GDP per capita,
  \item $W_i$ is a high-dimensional vector of control variables, including proxies for human capital, institutional quality, demographic structure, trade openness, and geographic characteristics (here, $p = 60$),
  \item and $n = 90$ denotes the number of countries in the sample.
\end{enumerate}

\bigskip
In low-dimensional settings, ordinary least squares (OLS) would be the default estimator. However, when the number of controls $p$ is comparable to or exceeds the sample size $n$, OLS becomes unstable and prone to overfitting. Moreover, the inclusion of many irrelevant controls inflates variance, resulting in wide confidence intervals and reduced statistical power. Double Lasso, on the other hand, reduces variance inflation and improves inference validity in high-dimensional settings.

\begin{center}
\begin{tabular}{@{}lccc@{}}
\toprule
\textbf{Method} & \textbf{Estimate} & \textbf{SE} & \textbf{95\% CI} \\
\midrule
OLS            & $-0.009$ & 0.032 & $[-0.073,\;0.054]$ \\
Double Lasso   & $-0.045$ & 0.018 & $[-0.080,\;-0.010]$ \\
\bottomrule
\end{tabular}
\end{center} The OLS estimate is close to zero and statistically insignificant, as its confidence interval includes zero. This reflects the high noise-to-signal ratio in the presence of many controls relative to sample size. In contrast, the Double Lasso estimate is larger in magnitude and statistically significant at conventional levels. Its tighter confidence interval excludes zero, offering empirical support for the hypothesis of conditional convergence.

\bigskip
The negative and significant estimate of $\alpha$ under Double Lasso suggests that, holding constant various institutional and structural characteristics, countries with lower initial GDP indeed tend to grow faster. This aligns with the predictions of the \href{https://en.wikipedia.org/wiki/Solow%E2%80%93Swan_model}{Solow growth model}.

\section{Beyond a Single Treatment}

Researchers often wish to infer a vector
\(\bm\alpha=(\alpha_1,\dots,\alpha_{p_1})\) attached to multiple policy
variables \(D_1,\dots,D_{p_1}\).  Double Lasso generalizes by residualizing
\emph{each} \(D_\ell\) on the control set \(W\) and running separate low‑dimensional
regressions.  Simultaneous confidence bands may be constructed via the Gaussian
multiplier bootstrap, controlling family‑wise error in the presence of
many targets.

Applications include interaction effects (policy $\times$ industry),
non‑linear polynomial expansions ($D$, $D^2$, $D^3$), and heterogeneity by
demographics.

\section{Alternative Orthogonal Methods}

An alternative to residualization in partialling-out is the \textbf{Double Selection} method. In this approach, two separate Lasso regressions are performed:

\begin{enumerate}
  \item Regress the outcome \( Y \) on the full set of covariates \( (D, W) \) using Lasso. Let \( \widehat{S}_Y \subseteq \{1,\ldots,p\} \) denote the selected variables (excluding \( D \)).
  \item Regress the treatment variable \( D \) on the covariates \( W \) using Lasso. Let \( \widehat{S}_D \subseteq \{1,\ldots,p\} \) denote the selected variables.
\end{enumerate} The final model includes the treatment \( D \) and the union of selected controls \( \widehat{S} = \widehat{S}_Y \cup \widehat{S}_D \). An ordinary least squares (OLS) regression of \( Y \) on \( D \) and the selected variables \( W_{\widehat{S}} \) is then performed to estimate the treatment effect. This estimator:
\begin{enumerate}
  \item Retains the \textbf{Neyman orthogonality} property,
  \item Is robust to model selection errors in either the outcome or treatment model,
  \item is typically more conservative—favoring inclusion of relevant covariates—than single-selection or naive Lasso.
\end{enumerate} Another method is The \textbf{Debiased Lasso}, also known as the \textbf{Desparsified Lasso}, addresses the regularization bias inherent in standard Lasso estimates. It constructs an estimator with an asymptotic linear expansion by correcting the shrinkage in the Lasso estimate. The procedure has been outlined below.
\begin{enumerate}
    \item Run a Lasso regression of \(Y\) on \(D\) and \(W\) (with an appropriate penalty \(\lambda\)) to obtain \(\hat{\beta}\).
    \item Run a Lasso regression of \(D\) on \(W\) (with a suitable \(\lambda\)) to obtain \(\hat{\gamma}\).
    \item Construct the residualized treatment:
    \[
    \tilde{D}(\hat{\gamma}) = D - \hat{\gamma}'W.
    \]
    \item Obtain the estimator \(\hat{\alpha}\) as the solution to the moment condition:
    \[
    \frac{1}{n}\sum_{i=1}^n \left( Y_i - \hat{\alpha} D_i - \hat{\beta}'W_i \right)\tilde{D}_i(\hat{\gamma}) = 0.
    \]
    This yields the explicit formula:
    \[
    \hat{\alpha} = \left( \frac{1}{n}\sum_{i=1}^n D_i\,\tilde{D}_i(\hat{\gamma}) \right)^{-1} \left( \frac{1}{n}\sum_{i=1}^n \left(Y_i-\hat{\beta}'W_i\right)\tilde{D}_i(\hat{\gamma}) \right).
    \]
\end{enumerate} The debiased estimator for a target coefficient \( \alpha \) has the form:
\[
\hat\alpha
  = \hat\alpha_{\text{Lasso}}
    + \frac{1}{n}\sum_{i=1}^{n} \hat\Theta^{\top} X_i
      \left(Y_i - X_i^{\top} \hat\beta_{\text{Lasso}}\right),
\]
where:
\begin{enumerate}
  \item \( \hat\beta_{\text{Lasso}} \) is the Lasso estimate from regressing \( Y \) on the full covariate vector \( X \),
  \item \( \hat\Theta \) is an estimate of the relevant row of the (pseudo-)inverse of the empirical Gram matrix \( \widehat{\Sigma} = X^{\top}X/n \).
\end{enumerate} This adjustment removes the bias from penalization and restores asymptotic normality:
\[
\sqrt{n}(\hat\alpha - \alpha) \xrightarrow{d} N(0, \sigma^2),
\]
even in high-dimensional settings where \( p \gg n \). Confidence intervals and hypothesis tests constructed from this estimator are asymptotically valid.

\bigskip
In the special case of a pure randomized controlled trial (RCT), the treatment \( D \) is independent of the covariates \( W \) by design:
\[
D \perp\!\!\!\perp W.
\]
Therefore, the naive OLS regression of \( Y \) on \( D \) yields an unbiased estimate of the average treatment effect (ATE), even without adjusting for \( W \). While orthogonal methods like Double Lasso or Debiased Lasso are not necessary to ensure unbiasedness in this setting, they can still be employed to increase estimation efficiency (i.e., reduce variance) and improve precision of inference when \( W \) explains a substantial amount of variation in \( Y \). While randomization solves the identification problem, orthogonal adjustments still play a role in optimizing statistical performance.

\begin{enumerate}
    \item Run a Lasso regression of \(Y\) on \(D\) and \(W\) (with an appropriate penalty \(\lambda\)) to obtain \(\hat{\beta}\).
    \item Run a Lasso regression of \(D\) on \(W\) (with a suitable \(\lambda\)) to obtain \(\hat{\gamma}\).
    \item Construct the residualized treatment:
    \[
    \tilde{D}(\hat{\gamma}) = D - \hat{\gamma}'W.
    \]
    \item Obtain the estimator \(\hat{\alpha}\) as the solution to the moment condition:
    \[
    \frac{1}{n}\sum_{i=1}^n \left( Y_i - \hat{\alpha} D_i - \hat{\beta}'W_i \right)\tilde{D}_i(\hat{\gamma}) = 0.
    \]
    This yields the explicit formula:
    \[
    \hat{\alpha} = \left( \frac{1}{n}\sum_{i=1}^n D_i\,\tilde{D}_i(\hat{\gamma}) \right)^{-1} \left( \frac{1}{n}\sum_{i=1}^n \left(Y_i-\hat{\beta}'W_i\right)\tilde{D}_i(\hat{\gamma}) \right).
    \]
\end{enumerate} Equivalently, the debiased estimator for a target coefficient \( \alpha \) has the form:
\[
\hat\alpha
  = \hat\alpha_{\text{Lasso}}
    + \frac{1}{n}\sum_{i=1}^{n} \hat\Theta^{\top} X_i
      \left(Y_i - X_i^{\top} \hat\beta_{\text{Lasso}}\right),
\]
where:
\begin{enumerate}
  \item \( \hat\beta_{\text{Lasso}} \) is the Lasso estimate from regressing \( Y \) on the full covariate vector \( X \),
  \item \( \hat\Theta \) is an estimate of the relevant row of the (pseudo-)inverse of the empirical Gram matrix \( \widehat{\Sigma} = X^{\top}X/n \).
\end{enumerate} This correction term serves to "undo" the bias introduced by Lasso’s penalization.

\section{Practical Checklist for Empirical Implementation}
\label{subsec:checklist}

When implementing orthogonal estimation methods such as Double Lasso or Debiased Lasso, it is important to adhere to a set of practical best practices that ensure numerical stability, valid inference, and robust results. Below is a detailed checklist to guide empirical implementation:

\begin{enumerate}
  \item \textbf{Centre and Scale Regressors.}\\
  Always standardize the covariates in the high-dimensional matrix \( W \). This involves centering each column to have mean zero and scaling to unit variance. Standardization:
  \begin{enumerate}
    \item Improves numerical stability and convergence in optimization,
    \item Ensures that the Lasso penalty is applied uniformly across variables,
    \item Makes interpretation of coefficients more meaningful, especially when variables are on different scales.
  \end{enumerate}

  \item \textbf{Penalty Selection (\( \lambda \)).}\\
  Choose the regularization parameter \( \lambda \) via one of two common methods:
  \begin{enumerate}
    \item \emph{10-fold Cross-Validation:} Splits the data into 10 parts, trains on 9 and validates on 1, cycling through all folds.
    \item \emph{Theory-Driven Plug-in Method:} Based on theoretical guarantees that balance false inclusion and exclusion probabilities (e.g., Belloni et al. plug-in).
  \end{enumerate}
  Avoid under-penalization, which can lead to overfitting and poor generalization.

  \item \textbf{Inference: Robust Standard Errors.}\\
  Since weights and nuisance parameters are estimated, standard error estimation in the second-stage regression must account for this:
  \begin{enumerate}
    \item Use heteroskedasticity-robust \emph{sandwich estimators},
    \item Or apply the \emph{residual bootstrap} to obtain valid confidence intervals in small samples,
    \item If data are clustered (e.g., by region or school), use clustered standard errors to account for intra-cluster correlation.

  \end{enumerate}

  \item \textbf{Diagnostic Plots.}\\
  Visual inspection of residuals and influential observations can prevent misinterpretation. Recommended diagnostics include:
  \begin{enumerate}
    \item Histogram of \( \check D_i \): the partialled-out regressor, to verify variation,
    \item Leverage scores: to identify high-leverage points,
    \item Cook’s distance: to detect influential observations that could distort estimates.
  \end{enumerate}

  \item \textbf{Sensitivity Analysis.}\\
  Examine robustness of the estimated treatment effect \( \hat\alpha \) by varying the penalty parameter \( \lambda \) by \( \pm 25\% \). If results change dramatically, your model may be fragile or overfit. Stable estimates across a penalty range provide greater empirical confidence.
\end{enumerate}
\section{Notebooks}  
The Jupyter notebooks for the referenced experimentation are available here in \href{https://github.com/Gauranga2022/CMI-MSc-Data-Science/tree/main/Sem4/CausalInference%20(GP)}{Github.}

\section{References}
\begin{enumerate}
  \item Victor Chernozhukov et al.\ (2023).  
        \emph{Applied Causal Inference Powered by Machine Learning and AI},
        Chapter 2.  
        \url{https://causalml-book.org/}
\end{enumerate}
\newpage

\chapter[Metrics to Recover Causal Effects]{Conditioning and Propensity Scores Recover Causal Effects}

In observational studies, the treatment indicator $D$ is almost never randomly assigned. Patients choose therapies, students select majors, and firms adopt technologies according to incentives, preferences, or constraints—factors that often \emph{also} influence the outcome $Y$ we intend to study. The resulting \textbf{selection bias} makes naive comparisons unreliable: average wage differences between trained and untrained workers reflect both training benefits and pre‑existing ability gaps. To address this challenge, two conceptual tools are central:

\begin{enumerate}
  \item \textbf{Conditional ignorability} (also known as \emph{unconfoundedness} or \emph{selection on observables}): This assumption posits that after conditioning on a sufficiently rich set of covariates $X$, the treatment assignment $D$ behaves \emph{as if} it were random. That is, any systematic differences in potential outcomes between treated and untreated units can be fully explained by $X$.
  
  \item \textbf{Overlap} (or \emph{positivity}): This ensures that, for every covariate profile $X$, there is a positive probability of observing both treated and untreated units. Formally, $0 < P(D = 1 \mid X) < 1$ almost surely. Without this condition, comparisons across treatment groups become ill-posed or undefined for certain strata.
\end{enumerate} When both of these conditions hold, researchers can recover the \emph{Average Treatment Effect} (ATE) using standard regression techniques or propensity-score reweighting—even without the benefit of randomized experiments. What follows is a self-contained exposition of how these tools work in practice.

\section{Potential Outcomes Framework}
\label{subsec:potential-outcomes}

In causal inference, the potential outcomes framework provides a rigorous way to define and estimate causal effects. We first set up the core notations that enable the identification of causal quantities from observational data and revisit a few assumptions.

\paragraph{Notation and Potential Outcomes.}

\begin{enumerate}
  \item $D$: A binary treatment indicator where $D = 1$ denotes treatment and $D = 0$ denotes control.
  \item $Y$: The observed outcome variable.
  \item $Y(1),Y(0)$: Potential outcomes—$Y(1)$ is the outcome that would be observed if the unit were treated, and $Y(0)$ is the outcome if the unit were not treated. These together define the \emph{response function} $d \mapsto Y(d)$ for $d \in \{0,1\}$.
  \item $X$: A vector of observed pre-treatment covariates (e.g., age, GPA, income, health status, etc.).
\end{enumerate} We assume the \textbf{Consistency} axiom, which states:
\[
Y = Y(D), \qquad \text{almost surely},
\]
meaning we observe the potential outcome corresponding to the realized treatment $D$.

\paragraph{Assumption of Conditional Ignorability:} For each $d \in \{0,1\}$, treatment assignment is independent of the potential outcome given covariates:
\[
D \;\perp\!\!\!\perp\; Y(d) \mid X.
\] To interpret, within covariate strata $X = x$, the treatment assignment behaves as if randomized. That is, the treated and untreated subpopulations are statistically comparable with respect to their potential outcomes. Any observed difference in outcomes between treated and control units sharing the same covariate values can thus be attributed to the effect of the treatment itself. For example, suppose two students have the same SAT score, family income, and high school quality (i.e., same $X$). If one attends college ($D=1$) and the other does not ($D=0$), conditional ignorability asserts that these students' potential earnings $Y(1), Y(0)$ are unrelated to the college attendance decision—once we have conditioned on $X$:= SAT score, family income, and high school quality.
\bigskip

\paragraph{Assumption of Overlap (Common Support):}

Let $p(X) = P(D = 1 \mid X)$ denote the \textbf{propensity score}, the probability of receiving treatment given covariates. The overlap condition requires:
\[
0 < p(X) < 1 \quad \text{almost surely}.
\] This ensures that every covariate profile $X$ has a positive probability of being assigned to both treatment and control groups. Without this assumption, there would be regions of $X$ where treated or untreated units are completely absent, making comparisons and causal effect estimation ill-posed.

\bigskip

Under these assumptions of Conditional Ignorability and Overlap, we can eliminate selection bias by conditioning on $X$:

\begin{align*}
E[Y \mid D = d, X]
  &= E[Y(d) \mid D = d, X] \quad && \text{(by Consistency)} \\
  &= E[Y(d) \mid X] \quad && \text{(by Ignorability)}.
\end{align*} Therefore, the \textbf{Conditional Average Predictive Effect (CAPE)} given by:
\[
\pi(X) = E[Y \mid D = 1, X] - E[Y \mid D = 0, X],
\] equals the \textbf{Conditional Average Treatment Effect (CATE)}:
\[
\delta(X) = E[Y(1) - Y(0) \mid X].
\] and averaging over the covariate distribution yields the \textbf{Average Treatment Effect (ATE)}:
\[
\delta = E[\delta(X)] = E[\pi(X)].
\] So under the assumptions of ignorability and overlap, causal effects can be identified purely from observational data using regression or weighting methods that condition on $X$, which we check next.

\section{Linear Regression as a Conditioning Device}
\label{subsec:linear-regression}

A standard approach to causal inference under the potential outcomes framework is to use linear regression as a tool for conditioning on observed covariates. This approach becomes valid under the assumption of \textbf{conditional ignorability}.

\bigskip

Suppose we posit the following linear model for the conditional expectation of the outcome:
\begin{equation}
\label{eq:lin}
E[Y \mid D, X] = \alpha D + \beta^{\!\top} X,
\end{equation}
where
\begin{enumerate}
  \item \( \alpha \) represents the average causal effect of the treatment, D
  \item \( \beta \in \mathbb{R}^p \) is a vector of coefficients associated with the covariates \( X \).
\end{enumerate} We then estimate the following linear regression model via ordinary least squares (OLS):
\[
Y = \alpha D + \beta^{\!\top} X + \varepsilon,
\]
where \( \varepsilon \) is a zero-mean error term satisfying \( E[\varepsilon \mid D, X] = 0 \).
Under the assumptions of linearity, ignorability, and overlap, the OLS estimator \( \hat{\alpha} \) is a consistent estimator of the \textbf{Average Treatment Effect} (ATE), denoted:
\[
\delta := E[Y(1) - Y(0)].
\]

\bigskip

To account for heterogeneous treatment effects, we can extend the linear model to allow the treatment effect to vary with covariates by introducing interaction terms between \( D \) and a centered version of the covariates \( X \). The modified model is:
\[
E[Y \mid D, X] =
\alpha_1 D + \alpha_2^{\!\top}(X - \bar{X}) D
+ \beta_1 + \beta_2^{\!\top}(X - \bar{X}),
\]
where:
\begin{enumerate}
  \item \( \bar{X} \) is the sample mean of the covariates,
  \item \( \alpha_1 \) captures the average treatment effect across all covariate strata (ATE),
  \item \( \alpha_2 \) allows the treatment effect to vary with characteristics.
\end{enumerate} In this case, we interpret:
\[
\text{ATE} = \alpha_1,
\qquad
\text{CATE}(X) = \alpha_1 + \alpha_2^{\!\top}(X - \bar{X}).
\]

\bigskip

The validity of this linear regression approach depends on the correctness of the specified functional form. If the linear model in Equation~\eqref{eq:lin} does not represent the true conditional expectation function (CEF), then the resulting estimate \( \hat{\alpha} \) may be biased. When the functional form is in doubt or the relationship between \( Y \), \( D \), and \( X \) is highly non-linear, it is advisable to turn to more flexible estimation techniques such as non-parametric regression or modern machine learning methods (e.g., random forests, boosted trees, neural networks). These approaches can better capture complex patterns in the data and help avoid misspecification bias.

\section{The Horvitz--Thompson Method}

In observational studies, the identification of treatment effects often relies on conditioning on observed covariates to model the \emph{outcome process}, i.e., the conditional expectation function
\[
E[Y \mid D, X].
\]
However, in many real-world applications, this function may be complex and difficult to approximate accurately. As an alternative, we may focus on the \emph{treatment assignment process}, specifically the \textbf{propensity score}, which is defined as
\[
p(X) = P(D = 1 \mid X).
\]
The Horvitz--Thompson (HT) method leverages this score to reweight observed outcomes, providing an alternative path to identification of potential outcome means. This approach is especially effective when the covariate vector \(X\) is high-dimensional, and the propensity score is either known (e.g., from design) or can be reliably estimated.

\bigskip
\def\theorem{\par\noindent{\bf Theorem.\ } \ignorespaces}
\def\endtheorem{}
\begin{theorem}
    Horvitz--Thompson Propensity Score Reweighting):
    Under the assumptions of Conditional Ignorability (\( Y(d) \perp D \mid X \)) and Overlap (\(0 < P(D = d \mid X) < 1\)), the conditional expectation of a properly reweighted observed outcome \(Y\) identifies the conditional mean of the potential outcome \(Y(d)\):
    \[
    E\left[\frac{Y \cdot 1(D = d)}{P(D = d \mid X)} \,\Big|\, X\right] = E[Y(d) \mid X].
    \]
    Taking expectation over the distribution of \(X\), we recover the marginal potential outcome mean:
    \[
    E\left[\frac{Y \cdot 1(D = d)}{P(D = d \mid X)}\right] = E[Y(d)].
    \]
\end{theorem}

\begin{proof}
    We apply the law of iterated expectations:
    \begin{align*}
    E\left[\frac{Y \cdot 1(D = d)}{P(D = d \mid X)} \,\Big|\, X\right]
    &= E\left[\frac{Y(d) \cdot 1(D = d)}{P(D = d \mid X)} \,\Big|\, X\right] \quad \text{(by consistency: } Y = Y(d) \text{ if } D = d) \\
    &= \frac{1}{P(D = d \mid X)} \, E\left[Y(d) \cdot 1(D = d) \mid X\right] \\
    &= \frac{1}{P(D = d \mid X)} \left(E[Y(d) \mid X, D = d] \cdot P(D = d \mid X)\right) \\
    &= E[Y(d) \mid X] \quad \text{(by ignorability: } Y(d) \perp D \mid X \text{).}
    \end{align*}
    Averaging both sides over \(X\) yields the identification of the average potential outcome \(E[Y(d)]\).
\end{proof}

\bigskip
As a direct application, we can define a transformed variable, the \textbf{Horvitz--Thompson Transform}:
\[
H = \frac{1\{D = 1\}}{P(D = 1 \mid X)} - \frac{1\{D = 0\}}{P(D = 0 \mid X)},
\] and recover the \textbf{Average Treatment Effect (ATE)} via:
\[
\delta = E[Y \cdot H].
\] This formulation ensures that each individual's observed outcome \(Y\) is weighted inversely by their probability of receiving the treatment they actually received, thus correcting for the non-random assignment. Similarly, we define the \textbf{Conditional Average Treatment Effect (CATE)} as
\[
\delta(X) = E[Y \cdot H \mid X].
\]
This formulation highlights how causal effects can be recovered both at the population level and at the covariate level, solely via reweighting based on estimated or known treatment probabilities.


\section{A Remark on Precision and Propensity Scores}

Suppose we are studying the effect of a job training program on wages. The propensity score \( p(X) = P(D=1 \mid X) \) may be estimated using various individual-level covariates such as age, education, and prior work experience. Now imagine two individuals who share the same estimated propensity score. This implies that they have an equal predicted probability of receiving the training, given their observed covariates. However, this similarity in the overall score may mask very different covariate profiles. For example:

\begin{enumerate}
    \item One individual may have substantial work experience and little formal education.
    \item The other may be highly educated but lack practical experience.
\end{enumerate} Even though their propensity scores are the same, their background differences might influence the outcome variable (e.g., wages) in distinct ways. This situation highlights a key limitation of relying solely on propensity scores for identification. Solution lies in Double Machine Learning. By combining both:
\begin{enumerate}
    \item Reweighting (based on the propensity score),
    \item Regression adjustment (using the full covariate vector \(X\)),
\end{enumerate} we can more effectively \emph{de-noise} the outcome variable. That is, we reduce residual variance by accounting for detailed covariate patterns. This hybrid approach—central to \textbf{double machine learning}—produces more efficient and precise estimates of the true treatment effect by leveraging both the treatment selection mechanism and the outcome-generating process.

\bigskip

\section{Covariate Balance Checks}

To assess whether randomization (or reweighting via the propensity score) has succeeded in balancing treatment and control groups, we perform a \textbf{covariate balance check}. Under the assumption of \textbf{conditional ignorability}, it must be that:
\[
E[H \mid X] = 0,
\]
where \( H = \frac{1\{D=1\}}{p(X)} - \frac{1\{D=0\}}{1-p(X)} \) is the Horvitz–Thompson transform. If this expectation does not hold—i.e., if some function of \(X\) is predictive of \(H\)—then systematic differences remain between the treatment and control groups even after reweighting. This implies a violation of ignorability, as treatment assignment remains confounded.

\bigskip
In a low-dimensional setting, we can test balance via regression. Specifically:

\begin{enumerate}
    \item Construct a dictionary \(W = f(X)\), containing linear and non-linear transformations of covariates.
    \item Regress \(H\) on \(W\).
    \item Perform an \(F\)-test to assess whether \(W\) significantly predicts \(H\).
\end{enumerate} If the regression shows predictive power (i.e., \(W\) significantly predicts \(H\)), then covariate imbalance exists—suggesting that the treatment assignment was not successfully randomized or the propensity score model was misspecified. Covariate balance diagnostics are essential validity checks for any causal analysis relying on ignorability. They are especially critical in observational studies where treatment is not randomly assigned.

\section{Directed Acyclic Graphs}

To understand how treatment assignment and potential outcomes are connected, we use \emph{Directed Acyclic Graphs (DAGs)}. These help us visualize the assumptions behind causal inference, particularly the assumption of \textbf{conditional ignorability}.

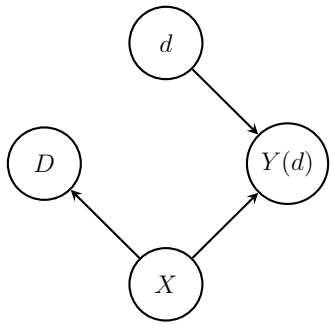
\begin{center}
\begin{tikzpicture}[scale=0.8, transform shape, ->, >=stealth, thick]
    \node[draw, circle, minimum size=1.2cm] (D) at (-2,0) {\(D\)};
    \node[draw, circle, minimum size=1.2cm] (d) at (0,2) {\(d\)};
    \node[draw, circle, minimum size=1.2cm] (Y) at (2,0) {\(Y(d)\)};
    \node[draw, circle, minimum size=1.2cm] (X) at (0,-2) {\(X\)};
    \draw (X) -- (D);
    \draw (X) -- (Y);
    \draw (d) -- (Y);
\end{tikzpicture}
\end{center}

\noindent \textbf{In this DAG:} 
\begin{enumerate}
    \item The node \(X\) represents observed pre-treatment covariates (e.g., age, education).
    \item The node \(D\) is the actual treatment assignment (e.g., whether the subject receives a drug).
    \item The node \(d\) denotes a hypothetical intervention assigning treatment status.
    \item \(Y(d)\) is the potential outcome under treatment level \(d\).
    \item The edge \(X \to D\) captures how baseline characteristics influence treatment choice.
    \item The edge \(X \to Y(d)\) encodes that these same characteristics may also affect the outcome.
    \item The edge \(d \to Y(d)\) reflects that the potential outcome depends on the treatment level.
\end{enumerate} For example in a healthcare context,
\begin{enumerate}
    \item \(X\): Patient's characteristics (e.g., age, comorbidities).
    \item \(d\): A hypothetical assignment to treatment (\(d=1\): drug, \(d=0\): no drug).
    \item \(Y(d)\): The patient's blood pressure if we hypothetically assign them treatment \(d\).
    \item \(D\): Actual treatment assignment observed in data.
    \item \(X \to D\) and \(X \to Y(d)\): Characteristics influence both treatment decision and health outcome.
    \item \(d \to Y(d)\): The potential outcome depends on which treatment is given.
    \item Under \textbf{ignorability}, once we condition on \(X\), treatment \(D\) behaves like a random assignment: outcome differences between groups with the same \(X\) can be causally attributed to treatment \(D\).
\end{enumerate}

We get the following as an Observed Outcome DAG: \(X \to D\), \(X \to Y\), \(D \to Y\).
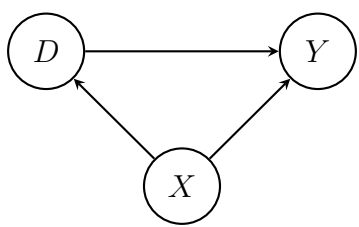
\begin{center}
\begin{tikzpicture}[->, >=stealth, thick, node distance=1.5cm]
    \node[draw, circle, minimum size=1cm] (X) {\(X\)};
    \node[draw, circle, minimum size=1cm, above left=of X] (D) {\(D\)};
    \node[draw, circle, minimum size=1cm, above right=of X] (Y) {\(Y\)};
    \draw (X) -- (D);
    \draw (X) -- (Y);
    \draw (D) -- (Y);
\end{tikzpicture}
\end{center}

This DAG depicts the factual world:
\begin{enumerate}
    \item Covariates \(X\) are realized first.
    \item Based on \(X\), each subject is assigned a treatment \(D\).
    \item The outcome \(Y\) is then determined by both the assigned treatment \(D\) and baseline covariates \(X\).
    \item Importantly, we only observe \(Y = Y(D)\), the outcome corresponding to the received treatment, not both potential outcomes.
\end{enumerate}

\section{Putting the Pieces Together}
To implement causal inference in observational studies, based on the above discussion, one must follow a structured approach combining both conceptual clarity and empirical rigor. The first step is \textbf{conceptual}, wherein the researcher articulates the assumptions of \emph{conditional ignorability} and \emph{overlap} within the context of the study and supports this reasoning with a directed acyclic graph (DAG) that visualizes the assumed causal structure. Next comes \textbf{empirical step A}, which involves estimating the treatment effect parameter $\alpha$ by fitting a regression model of the outcome variable $Y$ on the treatment indicator $D$ and covariates $X$. This regression can be linear or based on machine learning algorithms, and should be validated through model diagnostics such as goodness-of-fit and tests for functional form misspecification. \textbf{Empirical step B} focuses on the alternative strategy of weighting: the propensity score $\hat{p}(X)$ is estimated, Horvitz–Thompson (HT) weights $H$ are computed, and covariate balance is evaluated to ensure the weights have successfully removed systematic differences across treatment groups. Finally, in the \textbf{synthesis step}, results from both regression and weighting methods are compared. If the estimates align and covariate balance holds, one can credibly report estimates of the \emph{Average Treatment Effect} (ATE) and the \emph{Conditional Average Treatment Effect} (CATE). If not, the analyst should reconsider the underlying assumptions or explore alternative identification strategies such as instrumental variable methods.

\section*{Recommended Readings}

\begin{enumerate}
  \item Rosenbaum, P. R., \& Rubin, D. B. (1983). ``The Central Role of the Propensity Score in Observational Studies for Causal Effects.'' \emph{Biometrika}, \textbf{70}(1), 41–55.
  
  \item Imbens, G. W., \& Rubin, D. B. (2015). \emph{Causal Inference for Statistics, Social, and Biomedical Sciences: An Introduction}. Cambridge University Press.
  
  \item Hernán, M. A., \& Robins, J. M. (2020). \emph{Causal Inference: What If?} Boca Raton: Chapman \& Hall/CRC. Available at: \url{https://miguelhernan.org/whatifbook}
  
  \item Chernozhukov, V., Wüthrich, K., Kumar, M., Semenova, V., Yadlowsky, S., \& Zhu, Y. (2023). \emph{Applied Causal Inference Powered by Machine Learning and AI}. Available at: \url{https://causalml-book.org/}
\end{enumerate}

\chapter[Graphical Models]{Structured Equations Modelling and Graphical Models}

Statistical data alone inform us about associations, but causal science asks \emph{what would happen if we acted differently}. \textbf{Directed acyclic graphs} (DAGs) encode qualitative subject‑matter knowledge—who can influence whom—while \textbf{structural equation models} (SEMs) supply quantitative functional relationships that generate the joint distribution. Pearl (2000) showed that every (acyclic) SEM induces a DAG and, conversely, that the graph together with independent disturbances suffices to reconstruct
counterfactual outcomes. We begin by exploring a fully linear and nonlinear, nonparametric formulation of causal diagrams and their associated structural equation models (SEMs). These models serve as powerful and flexible tools for uncovering the underlying structure necessary for causal identification, enabling us to move beyond the confines of purely linear assumptions. Within this framework, we define counterfactuals in a formal way—adhering to what Judea Pearl describes as the "First Law of Causal Inference"—which states that every structural equation model naturally induces a system of counterfactual outcomes.

\section{General DAG and SEM via an Example}

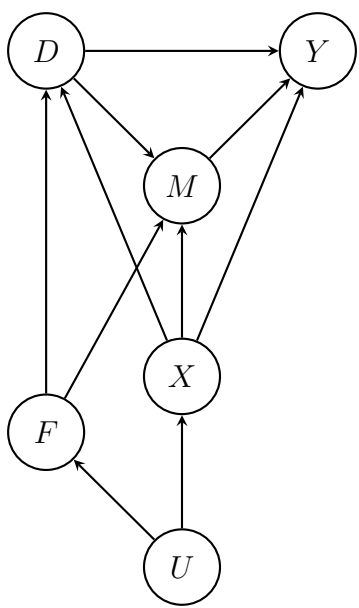
\begin{center}
\begin{tikzpicture}[->, >=stealth, thick, node distance=1.5cm]
    \node[draw, circle, minimum size=1cm] (M) {\(M\)};
    \node[draw, circle, minimum size=1cm, above left=of M] (D) {\(D\)};
    \node[draw, circle, minimum size=1cm, above right=of M] (Y) {\(Y\)};
    \node[draw, circle, minimum size=1cm, below=of M] (X) {\(X\)};
    \node[draw, circle, minimum size=1cm, below=of X] (U) {\(U\)};
    \node[draw, circle, minimum size=1cm, above left=of U] (F) {\(F\)};
    
    \draw (D) -- (M);
    \draw (D) -- (Y);
    \draw (M) -- (Y);
    \draw (U) -- (F);
    \draw (U) -- (X);
    \draw (X) -- (M);
    \draw (X) -- (Y);
    \draw (X) -- (D);
    \draw (F) -- (M);
    \draw (F) -- (D);
\end{tikzpicture}
\end{center} We now illustrate a causal diagram and its associated structural equation model (SEM) using a real-world scenario: the impact of 401(k) eligibility on financial wealth. The directed acyclic graph (DAG) below encodes the causal relationships among various observed and unobserved variables involved in this context.

In the United States, a 401(k) plan is an employer-sponsored, defined-contribution, personal pension (savings) account, as defined in section 401(k) of the U.S. Internal Revenue Code. This causal graph represents the possible channels through which 401(k) eligibility (\(D\)) may affect an individual's net financial wealth (\(Y\)). The interpretation of each node is as follows:

\begin{enumerate}
    \item \(D\): Binary indicator for 401(k) eligibility — the treatment variable.
    \item \(Y\): Net financial assets — the outcome of interest.
    \item \(M\): Employer’s matching contribution — a potential mediator between eligibility and wealth.
    \item \(X\): Observed worker-level covariates such as age, income, and job tenure.
    \item \(F\): Unobserved firm-level covariates such as firm size or culture.
    \item \(U\): General latent factors, which may include unmeasured personality traits or risk preferences.
\end{enumerate} The arrows reflect assumed causal relationships. For example, \(D \rightarrow M\) and \(M \rightarrow Y\) represent the fact that eligibility for a 401(k) program may lead to a matching contribution, which in turn could influence net wealth. The variable \(X\) affects many nodes: it influences both treatment assignment (\(D\)) and outcome (\(Y\)), as well as the mediator (\(M\)). The unobserved variables \(U\) and \(F\) introduce latent confounding: \(U\) may simultaneously affect observed covariates \(X\) and firm characteristics \(F\), while \(F\) affects both \(D\) and \(M\).

\bigskip
This DAG can be translated into a structural system of equations where each node is generated as a function of its parent nodes and an associated noise term.

\[
\begin{aligned}
X &= f_X(U, \varepsilon_X), \\
F &= f_F(U, \varepsilon_F), \\
D &= f_D(X, F, \varepsilon_D), \\
M &= f_M(D, X, F, \varepsilon_M), \\
Y &= f_Y(D, M, X, \varepsilon_Y).
\end{aligned}
\] Each equation specifies how a variable is generated as a function of its direct causes (or "parents" in the DAG) and an associated idiosyncratic error term. For example, the covariates \(X\) are influenced by a latent factor \(U\), which could represent unobserved personal attributes such as financial literacy or long-term planning ability, along with a noise component \(\varepsilon_X\). Similarly, the firm-level variable \(F\) is also influenced by \(U\) and its own noise \(\varepsilon_F\), indicating that some unobserved traits may jointly affect both worker- and firm-level characteristics. The treatment assignment \(D\), indicating 401(k) eligibility, is determined by observed covariates \(X\), firm-level factors \(F\), and a residual \(\varepsilon_D\) capturing individual-level randomness in eligibility. The employer’s matching contribution \(M\) depends on \(D\), \(X\), and \(F\), suggesting that employer contributions may vary based on both employee and firm characteristics. Finally, the outcome of interest \(Y\), such as financial wealth, is modeled as a function of the treatment \(D\), mediator \(M\), covariates \(X\), and an independent shock \(\varepsilon_Y\). Together, these equations define a recursive system that captures the flow of causality in the model.

\section{Conditional Ignorability and Exogeneity}

A fundamental challenge in causal inference is to identify the conditions under which we can estimate causal effects from observed data. One such key condition is known as \textit{conditional ignorability} or \textit{conditional exogeneity}. This concept is central to modern causal analysis, as it provides the bridge between the structural assumptions encoded in a causal model and the statistical procedures used for estimation.

\bigskip

The connection between structural equation models (SEMs) and potential outcomes is foundational in causal inference. The fact that an SEM implies the existence of potential outcomes is sometimes called the \emph{First Law of Causal Inference}. SEMs, or their graphical representations as directed acyclic graphs (DAGs), encapsulate the contextual and substantive knowledge about the causal relationships in a given problem. As a result, they allow us to derive, rather than merely assume, important identification conditions such as conditional ignorability.

For example, suppose we are interested in the effect of a binary treatment \( D \) on an outcome \( Y \), and we have observed covariates \( F \) and \( X \). The SEM or DAG for the problem may suggest that, after conditioning on \( F \) and \( X \), the treatment assignment \( D \) is independent of the potential outcome \( Y(d) \) for each value of \( d \). Formally, this is written as:
\[
Y(d) \perp\!\!\!\perp D \mid F, X,
\]
which implies the following equality of conditional expectations:
\[
\mathbb{E}[Y(d) \mid F, X] = \mathbb{E}[Y \mid D = d, F, X].
\]
This property allows us to identify average causal (or treatment) effects by adjusting for (i.e., conditioning on) \( F \) and \( X \). Conditional ignorability (or exogeneity) can be justified using both functional (structural) arguments based on SEMs and graphical arguments based on d-separation and the backdoor criterion.

\bigskip

To provide a functional argument, consider the structural equations that define the data-generating process. In the counterfactual (or potential outcome) setting where we fix \( D = d \), the relevant structural equations are:
\[
Y(d) = f_Y(d, M(d), X, \epsilon_Y),
\]
\[
M(d) = f_M(d, F, X, \epsilon_M),
\]
where \( M \) is a mediator and \( \epsilon_Y, \epsilon_M \) are exogenous noise terms. The actual treatment assignment is generated by
\[
D = f_D(F, X, U, \epsilon_D),
\]
where \( U \) and \( \epsilon_D \) are additional exogenous variables. Once we condition on \( F \) and \( X \), the distribution of \( Y(d) \) is determined solely by \( d \), \( M(d) \), \( X \), and their associated noise terms. The realized value of \( D \) is not relevant for the distribution of \( Y(d) \) once \( F \) and \( X \) are given. In other words, knowing \( D \) provides no additional information about \( Y(d) \) beyond what is already known from \( F \) and \( X \):
\[
Y(d) \perp\!\!\!\perp D \mid F, X.
\]
This structural argument demonstrates how the ignorability condition can be justified by the functional relationships in the SEM.

\bigskip
The same conclusion can be reached using graphical criteria. In the counterfactual DAG, the node \( Y(d) \) receives inputs from \( M(d) \), \( X \), and the fixed value \( d \). The treatment variable \( D \) is still present in the graph and is generated by its usual parents (\( F \), \( X \), and \( U \)), but there is no direct arrow from \( D \) to \( Y(d) \).

Any path from \( D \) to \( Y(d) \) must pass through either \( F \) or \( X \). For instance, typical paths include:
\begin{enumerate}
    \item \( D \leftarrow X \to Y(d) \),
    \item \( D \leftarrow F \to M(d) \to Y(d) \),
    \item \( D \leftarrow F \leftarrow U \to X \to Y(d) \).
\end{enumerate}
By conditioning on \( F \) and \( X \), we block all such paths. In graphical terms, this is known as \emph{d-separation}: conditioning on a node severs the flow of information along any path passing through that node. The Global Markov property then tells us that d-separation implies conditional independence, to conclude that
\[
Y(d) \perp\!\!\!\perp D \mid F, X.
\]

\bigskip

A related graphical concept is the \emph{backdoor criterion}, which provides a practical method for identifying a set of variables \( Z \) that, when conditioned on, blocks all non-causal (backdoor) paths from \( D \) to \( Y \). A set \( Z \) satisfies the backdoor criterion if:
\begin{enumerate}
    \item No variable in \( Z \) is a descendant of \( D \), and
    \item \( Z \) blocks every backdoor path from \( D \) to \( Y \) (that is, every path that starts with an arrow into \( D \)).
\end{enumerate}
The first rule prevents us from blocking the causal effect of \( D \) on \( Y \), while the second ensures that all confounding paths are eliminated. In the context of the 401(k) example, the relevant backdoor paths from \( D \) to \( Y \) run through \( F \) and \( X \), such as:
\begin{enumerate}
    \item \( D \leftarrow X \to Y \),
    \item \( D \leftarrow F \to M \to Y \),
    \item \( D \leftarrow F \leftarrow U \to X \to Y \).
\end{enumerate}
By conditioning on both \( F \) and \( X \), we block all such paths, ensuring that the observed association between \( D \) and \( Y \) reflects only the causal effect of \( D \) on \( Y \).

\section{DAGs, SEMs, and d-Separation}

DAGs and their associated SEMs provide a precise mathematical language for representing causal systems. These concepts are essential for both the design and analysis of causal inference strategies. This section introduces the foundational concepts of directed acyclic graphs (DAGs), their associated structural equation models (SEMs), and the graphical criteria for conditional independence known as d-separation. These tools form the backbone of modern causal inference, providing a rigorous language for representing, analyzing, and reasoning about causal relationships.

\bigskip

A \textbf{directed acyclic graph (DAG)} is a graph \(G = (X, E)\) consisting of a set of nodes \(X = \{X_1, X_2, \dots ,X_{|X| }\}\) and directed edges \(E\), with the critical property that there are no cycles—meaning it is impossible to start at any node and follow a sequence of directed edges that eventually loops back to the starting node. Equivalently, the set of nodes \(V\) is partially ordered by the edge structure \(E\). Within a DAG, several relationships among nodes are defined:
\begin{enumerate}
    \item The \textbf{parents} of a node \(X_j\), denoted \(Pa_j\), are all nodes with directed edges pointing into \(X_j\):
    \[
    Pa_j := \{X_k : X_k \to X_j\}.
    \]
    \item The \textbf{children} of \(X_j\), denoted \(Ch_j\), are all nodes that \(X_j\) points to:
    \[
    Ch_j := \{X_k : X_j \to X_k\}.
    \]
    \item The \textbf{ancestors} of \(X_j\), denoted \(An_j\), are all nodes from which there exists a directed path to \(X_j\), including \(X_j\) itself:
    \[
    An_j := \{X_k : X_k < X_j\} \cup \{X_j\}.
    \]
    \item The \textbf{descendants} of \(X_j\), denoted \(Ds_j\), are all nodes that can be reached by a directed path starting from \(X_j\):
    \[
    Ds_j := \{X_k : X_k > X_j\}.
    \]
\end{enumerate} A DAG can be associated with an \textbf{acyclic structural equation model (ASEM)}, which formalizes how each variable is generated as a function of its parents and some exogenous noise. For each node \(j \in V\), the structural equation is:
\[
X_j := f_j(Pa_j, \epsilon_j),
\]
where each \(\epsilon_j\) is a random disturbance (exogenous variable), and the collection \((\epsilon_j)_{j \in V}\) is assumed to be jointly independent. A \textbf{linear ASEM} is a special case where the structural equations are linear in the parents:
\[
f_j(Pa_j, \epsilon_j) := f'_j Pa_j + \epsilon_j,
\]
with \(f'_j\) being a vector of coefficients. In linear ASEMs, the independence assumption on the errors can be weakened to mere uncorrelatedness. The \textbf{structural potential response process} for each variable describes how the value of \(X_j\) would respond to arbitrary assignments of its parents:
\[
X_j(pa_j) := f_j(pa_j, \epsilon_j),
\]
viewed as a stochastic process indexed by the possible values of the parents.

\bigskip

To generate observable variables in an ASEM, one draws a realization of the exogenous shocks \(\{\epsilon_j\}_{j \in V}\) and then solves the system of equations to obtain the endogenous variables \(\{X_j\}_{j \in V}\). The exogenous variables are those not determined by the model (the \(\epsilon_j\)), while the endogenous variables are those generated as outputs of the structural equations (the \(X_j\)).

\bigskip

A key property of ASEMs associated with DAGs is the \textbf{Markov factorization}, which states that the joint distribution of all variables factorizes as:
\[
p(\{x_\ell\}_{\ell \in V}) = \prod_{\ell \in V} p(x_\ell \mid pa_\ell).
\]
Equivalently, each variable is independent of its non-descendants given its parents—this is known as the \emph{local Markov property}.

Understanding the flow of information and potential confounding in DAGs requires precise definitions of paths and how they may be blocked. Now we define Paths, Blocked Paths, and d-Separation in DAGs
\begin{enumerate}
    \item A \textbf{directed path} is a sequence of nodes connected by edges all pointing in the same direction:
    \[
    X_{v_1} \to X_{v_2} \to \cdots \to X_{v_m}.
    \]
    \item A \textbf{non-directed path} allows some arrows to be reversed.
    \item A node \(X_j\) is a \textbf{collider} on a path if the path includes a segment of the form:
    \[
    \to X_j \leftarrow.
    \]
    \item A \textbf{backdoor path} from \(X_l\) to \(X_k\) is a non-directed path that starts at \(X_l\) and ends with an arrow into \(X_k\).
\end{enumerate} A path, \(\pi\) is said to be \textbf{blocked} by a set of nodes \(S\) if either:
\begin{enumerate}
    \item \(\pi\) contains a chain (\(i \to m \to j\)) or a fork (\(i \leftarrow m \to j\)) with \(m \in S\), or
    \item \(\pi\) contains a collider (\(i \to m \leftarrow j\)), and neither \(m\) nor any of its descendants are in \(S\).
\end{enumerate}
If a path is not blocked by \(S\), it is called \textbf{open}. Conditioning on a node in a chain or fork blocks the path, while conditioning on a collider or its descendant opens a path that would otherwise be blocked. For example, in the figures below, the backdoor path \(Y \leftarrow X \to D\) in (a) is blocked by conditioning on \(X\). Conversely, a path like \(Y \to C \leftarrow D\) in (b) is blocked by default, but becomes open if we condition on the collider \(C\).

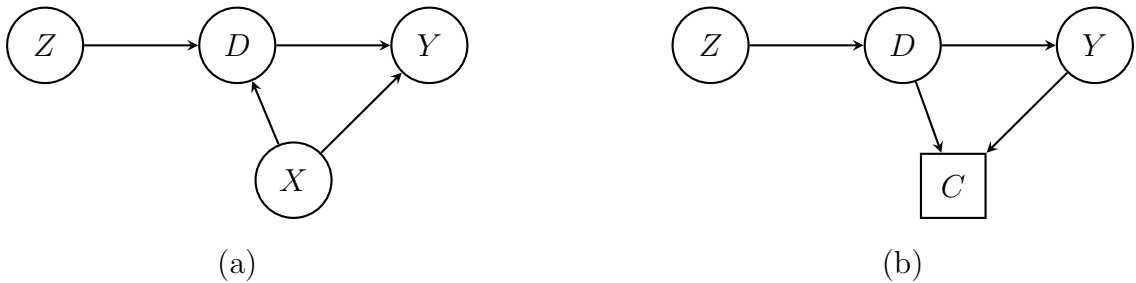
\begin{figure}[htbp]
    \centering
    \begin{minipage}{0.45\textwidth}
        \centering
        \begin{tikzpicture}[->, >=stealth, thick, node distance=1.5cm]
            \node[draw, circle, minimum size=1cm] (Z) {\(Z\)};
            \node[draw, circle, minimum size=1cm, right=of Z] (D) {\(D\)};
            \node[draw, circle, minimum size=1cm, right=of D] (Y) {\(Y\)};
            \node[draw, circle, minimum size=1cm, below left=of Y] (X) {\(X\)};
            \draw (Z) -- (D);
            \draw (D) -- (Y);
            \draw (X) -- (Y);
            \draw (X) -- (D);
        \end{tikzpicture}
        \par\vspace{0.5em}
        (a)
    \end{minipage}
    \hfill
    \begin{minipage}{0.45\textwidth}
        \centering
        \begin{tikzpicture}[->, >=stealth, thick, node distance=1.5cm]
            \node[draw, circle, minimum size=1cm] (Z) {\(Z\)};
            \node[draw, circle, minimum size=1cm, right=of Z] (D) {\(D\)};
            \node[draw, circle, minimum size=1cm, right=of D] (Y) {\(Y\)};
            \node[draw, rectangle, minimum size=0.85cm, below left=of Y] (C) {\(C\)};
            \draw (Z) -- (D);
            \draw (D) -- (Y);
            \draw (Y) -- (C);
            \draw (D) -- (C);
        \end{tikzpicture}
        \par\vspace{0.5em}
        (b)
    \end{minipage}
\end{figure} The concept of \textbf{d-separation} provides a graphical criterion for determining conditional independence in DAGs. Given a DAG \(G\), a set of nodes \(S\) is said to d-separate nodes \(X\) and \(Y\) if \(S\) blocks all paths between \(X\) and \(Y\). This is denoted as:
\[
(X \perp\!\!\!\perp_d Y | S)_G.
\]
By the foundational result of Pearl and Verma, d-separation implies conditional independence in the probability distribution generated by the DAG:
\[
X \perp\!\!\!\perp Y \mid S.
\]
Intuitively, if all information flow between \(X\) and \(Y\) is interrupted by conditioning on \(S\), then knowing \(X\) provides no additional information about \(Y\) once \(S\) is known. The formal proof of this equivalence is nontrivial, and the converse does not always hold.

\bigskip

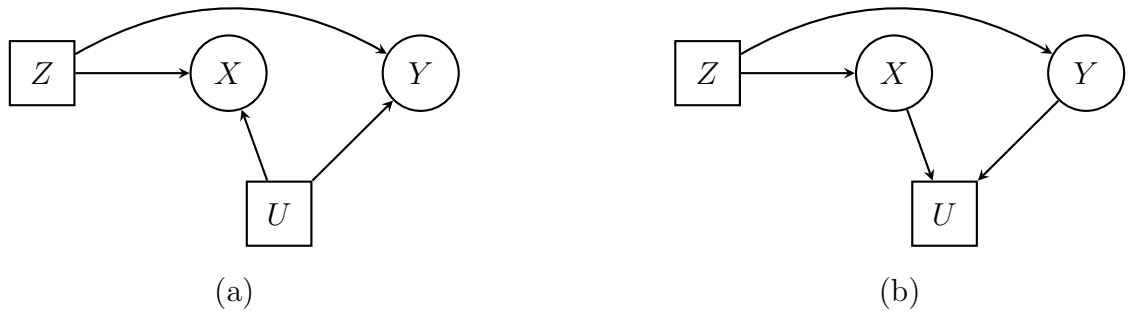
\begin{figure}[htbp]
    \centering
    \begin{minipage}[t]{0.45\textwidth}
        \centering
        \begin{tikzpicture}[->, >=stealth, thick, node distance=1.5cm]
            \node[draw, rectangle, minimum size=0.85cm] (Z) {\(Z\)};
            \node[draw, circle, minimum size=1cm, right=of Z] (X) {\(X\)};
            \node[draw, circle, minimum size=1cm, right=of X] (Y) {\(Y\)};
            \node[draw, rectangle, minimum size=0.85cm, below left=of Y] (U) {\(U\)};
            \draw (Z) -- (X);
            \draw (U) -- (X);
            \draw (U) -- (Y);
            \draw[bend left] (Z) to (Y);
        \end{tikzpicture}
        \par\vspace{0.5em}
        (a)
    \end{minipage}%
    \hfill
    \begin{minipage}[t]{0.45\textwidth}
        \centering
        \begin{tikzpicture}[->, >=stealth, thick, node distance=1.5cm]
            \node[draw, rectangle, minimum size=0.85cm] (Z) {\(Z\)};
            \node[draw, circle, minimum size=1cm, right=of Z] (X) {\(X\)};
            \node[draw, circle, minimum size=1cm, right=of X] (Y) {\(Y\)};
            \node[draw, rectangle, minimum size=0.85cm, below left=of Y] (U) {\(U\)};
            \draw (Z) -- (X);
            \draw (X) -- (U);
            \draw (Y) -- (U);
            \draw[bend left] (Z) to (Y);
        \end{tikzpicture}
        \par\vspace{0.5em}
        (b)
    \end{minipage}
\end{figure} 
To make these concepts concrete, consider the following two graphical structures below. In the first example, suppose the graph contains nodes \(Z\), \(U\), \(X\), and \(Y\), with edges such that \(Z\) and \(U\) are parents of \(X\), \(U\) is also a parent of \(Y\), and there is a direct edge from \(Z\) to \(Y\). Here, the set \(S = \{Z, U\}\) d-separates \(X\) and \(Y\), blocking all paths between them. The Markov factorization yields:
\[
p(y, x \mid u, z) = p(y \mid x, z, u) \, p(x \mid z, u) = p(y \mid u, z) \, p(x \mid z, u),
\]
implying that \(X \perp\!\!\!\perp Y \mid Z, U\).
In the second example, suppose the structure is such that \(Z\) points to \(X\), \(X\) and \(Y\) both point to \(U\), and \(Z\) also points to \(Y\). In this case, conditioning on \(Z\) alone d-separates \(X\) and \(Y\), and the factorization becomes:
\[
p(y, x \mid z) = p(y \mid z) \, p(x \mid z),
\]
which implies \(X \perp\!\!\!\perp Y \mid Z\).

\bigskip

\section{Intervention, Counterfactual DAGs, and SWIGs}

When analyzing causal effects, it is essential to understand how interventions modify the structure of a causal system. This is formalized through the concepts of counterfactual DAGs and Single World Intervention Graphs (SWIGs), which provide a unified graphical framework for reasoning about interventions and potential outcomes.

\bigskip

Consider a causal system represented by a directed acyclic graph (DAG) and its associated structural equation model (SEM). Each node in the DAG corresponds to a variable, and the directed edges represent causal relationships. To analyze the effect of an intervention—such as setting a treatment variable \(X_j\) to a specific value \(x_j\)—we must construct a new graphical object that reflects this hypothetical scenario. The intervention \(\text{fix}(X_j = x_j)\) transforms the original DAG into a counterfactual DAG, known as a Single World Intervention Graph (SWIG). The SWIG is constructed by a node-splitting operation:

\begin{enumerate}
    \item The original node \(X_j\) is split into two distinct entities:
    \begin{enumerate}
        \item \(X_j^*\), representing the natural (pre-intervention) value of \(X_j\).
        \item A new deterministic node, denoted \(a\) (or \(X_a^*\)), set to the intervened value \(x_j\).
    \end{enumerate}
    \item The intervention node \(X_a^*\) inherits only the outgoing edges from \(X_j\) (i.e., \(\widetilde{e}_{ai} = e_{ji}\) for all \(i\)) and has no incoming edges (\(\widetilde{e}_{ia} = 0\) for all \(i\)), reflecting that it is fixed by the intervention.
    \item The node \(X_j^*\) inherits only the incoming edges from \(X_j\) (i.e., \(\widetilde{e}_{ij} = e_{ij}\) for all \(i\)) and has no outgoing edges (\(\widetilde{e}_{ji} = 0\) for all \(i\)), preserving its dependence on its original causes.
    \item All remaining edges are preserved: \(\widetilde{e}_{ik} = e_{ik}\) for all \(i\) and for all \(k \neq j, k \neq a\), ensuring the rest of the graph structure remains intact.
    \item The counterfactual variables are assigned according to:
    \[
    X_k^* := f_k(Pa_k^*, \epsilon_k), \quad \text{for } k \neq a,
    \]
    where \(Pa_k^*\) denotes the parents of \(X_k^*\) under \(\widetilde{E}\), adapting the structural equations to the new graph.
\end{enumerate} The resulting SWIG, denoted \(\widetilde{G}(x_j)\), contains a set of counterfactual variables \(\{X_k^*\}_{k \in V} \cup \{X_a^*\}\), where each variable is now interpreted as a function of the intervention. The counterfactual SEM (CF-ASEM) associated with the SWIG defines the structural equations for the counterfactual variables:
\[
X_k^* := f_k(Pa_k^*, \epsilon_k), \quad \text{for } k \neq a,
\]
where \(Pa_k^*\) denotes the parents of \(X_k^*\) in the modified edge set, and \(\epsilon_k\) are the exogenous noise terms. This construction ensures that the counterfactual variables are generated consistently with the intervention. To illustrate, consider the following diagrams. The left figure shows the original DAG with variables and their causal relationships. The right figure shows the corresponding SWIG after intervening to set variable \(X_j\) to value \(x_j\):

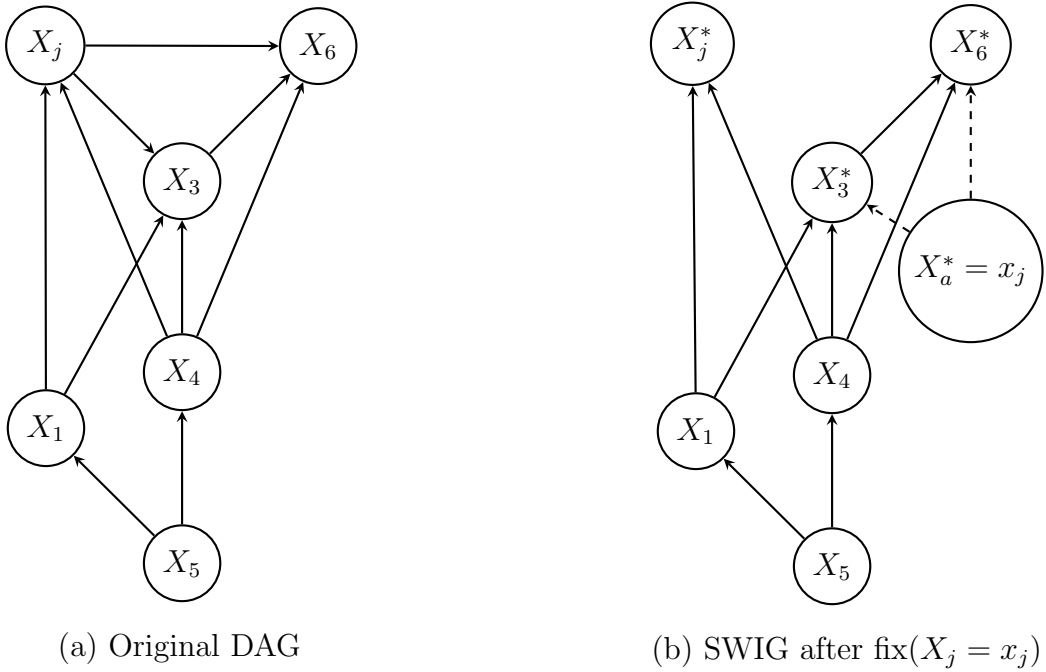
\begin{figure}[htbp]
    \centering
    \begin{minipage}{0.45\textwidth}
        \centering
        \begin{tikzpicture}[->,>=stealth, thick, node distance=1.5cm]
            \node[draw, circle, minimum size=1cm] (X3) {\(X_3\)};
            \node[draw, circle, minimum size=1cm, above left=of X3] (X2) {\(X_j\)};
            \node[draw, circle, minimum size=1cm, above right=of X3] (X6) {\(X_6\)};
            \node[draw, circle, minimum size=1cm, below=of X3] (X4) {\(X_4\)};
            \node[draw, circle, minimum size=1cm, below=of X4] (X5) {\(X_5\)};
            \node[draw, circle, minimum size=1cm, above left=of X5] (X1) {\(X_1\)};
            \draw (X2) -- (X3);
            \draw (X2) -- (X6);
            \draw (X3) -- (X6);
            \draw (X5) -- (X1);
            \draw (X5) -- (X4);
            \draw (X4) -- (X3);
            \draw (X4) -- (X6);
            \draw (X4) -- (X2);
            \draw (X1) -- (X3);
            \draw (X1) -- (X2);
        \end{tikzpicture}
        \par\vspace{0.5em}
        (a) Original DAG
    \end{minipage}
    \hfill
    \begin{minipage}{0.45\textwidth}
        \centering
        \begin{tikzpicture}[->,>=stealth, thick, node distance=1.5cm]
            \node[draw, circle, minimum size=1cm] (X3c) {\(X_3^*\)};
            \node[draw, circle, minimum size=1cm, above left=of X3c] (X2) {\(X_j^*\)};
            \node[draw, circle, minimum size=1cm, above right=of X3c] (X6c) {\(X_6^*\)};
            \node[draw, circle, minimum size=1cm, below=of X3c] (X4) {\(X_4\)};
            \node[draw, circle, minimum size=1cm, below=of X4] (X5) {\(X_5\)};
            \node[draw, circle, minimum size=1cm, above left=of X5] (X1) {\(X_1\)};
            \node[draw, circle, minimum size=1cm, below=of X6c] (X7) {\(X_a^* = x_j\)};
            \draw[dashed] (X7) -- (X3c);
            \draw[dashed] (X7) -- (X6c);
            \draw (X3c) -- (X6c);
            \draw (X5) -- (X1);
            \draw (X5) -- (X4);
            \draw (X4) -- (X3c);
            \draw (X4) -- (X6c);
            \draw (X4) -- (X2);
            \draw (X1) -- (X3c);
            \draw (X1) -- (X2);
        \end{tikzpicture}
        \par\vspace{0.5em}
        (b) SWIG after \(\text{fix}(X_j = x_j)\)
    \end{minipage}
\end{figure}

\bigskip

A central result is that the SWIG encodes the conditional independence relations among counterfactual variables under the intervention. Specifically, suppose we relabel the treatment node as \(D\), and let \(Y\) be any descendant of \(D\). Construct the SWIG induced by \(\text{fix}(D = d)\), and let \(S\) be any subset of nodes common to both the original DAG and the SWIG such that \(Y(d)\) is d-separated from \(D\) by \(S\) in the SWIG. Then:
\begin{enumerate}
    \item The conditional exogeneity condition holds:
    \[
    Y(d) \perp\!\!\!\perp D \mid S.
    \]
    \item The conditional average potential outcome is identified by:
    \[
    \mathbb{E}[g(Y(d)) \mid S = s] = \mathbb{E}[g(Y) \mid D = d, S = s],
    \]
    for all \(s\) with \(p(s, d) > 0\) and for all bounded functions \(g\).
\end{enumerate} To further clarify, consider the following example inspired by Pearl's work. The left diagram shows the original DAG, and the right shows the SWIG after intervening to set \(D = d\):

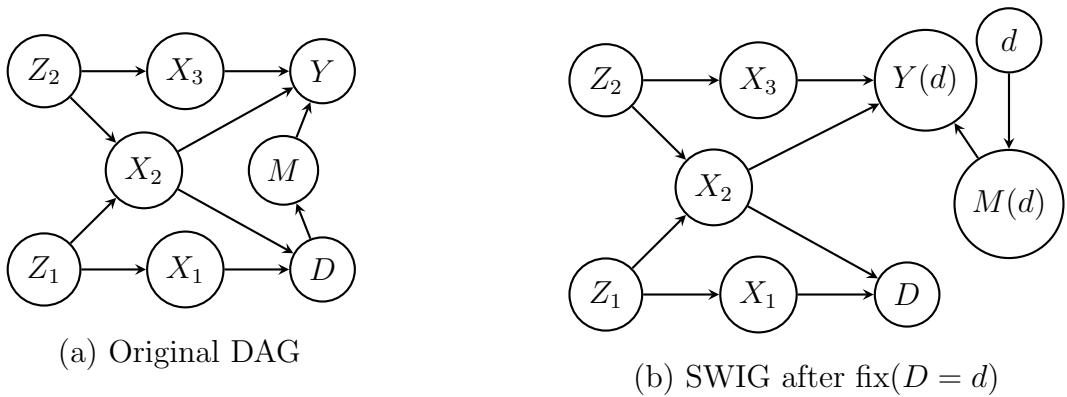
\begin{figure}[htbp]
    \centering
    \begin{minipage}{0.45\textwidth}
        \centering
        \begin{tikzpicture}[->, >=stealth, thick, node distance=0.85cm]
            \node[draw, circle, minimum size=0.85cm] (X_2) {\(X_2\)};
            \node[draw, circle, minimum size=0.85cm, right=of X_2] (M) {\(M\)};
            \node[draw, circle, minimum size=0.85cm, above left=of X_2] (Z_2) {\(Z_2\)};
            \node[draw, circle, minimum size=0.85cm, right=of Z_2] (X_3) {\(X_3\)};
            \node[draw, circle, minimum size=0.85cm, right=of X_3] (Y) {\(Y\)};
            \node[draw, circle, minimum size=0.85cm, below left=of X_2] (Z_1) {\(Z_1\)};
            \node[draw, circle, minimum size=0.85cm, right=of Z_1] (X_1) {\(X_1\)};
            \node[draw, circle, minimum size=0.85cm, right=of X_1] (D) {\(D\)};
            \draw (Z_2) -- (X_3);
            \draw (X_3) -- (Y);
            \draw (Z_2) -- (X_2);
            \draw (X_2) -- (Y);
            \draw (X_2) -- (D);
            \draw (D) -- (M);
            \draw (M) -- (Y);
            \draw (Z_1) -- (X_2);
            \draw (Z_1) -- (X_1);
            \draw (X_1) -- (D);
        \end{tikzpicture}
        \par\vspace{0.5em}
        (a) Original DAG
    \end{minipage}
    \hfill
    \begin{minipage}{0.50\textwidth}
        \centering
        \begin{tikzpicture}[->, >=stealth, thick, node distance=1.0cm]
            \node[draw, circle, minimum size=0.85cm] (X_2) {\(X_2\)};
            \node[draw, circle, minimum size=0.50cm, below right=of Y] (M) {\(M(d)\)};
            \node[draw, circle, minimum size=0.85cm, above left=of X_2] (Z_2) {\(Z_2\)};
            \node[draw, circle, minimum size=0.85cm, right=of Z_2] (X_3) {\(X_3\)};
            \node[draw, circle, minimum size=0.85cm, right=of X_3] (Y) {\(Y(d)\)};
            \node[draw, circle, minimum size=0.85cm, below left=of X_2] (Z_1) {\(Z_1\)};
            \node[draw, circle, minimum size=0.85cm, right=of Z_1] (X_1) {\(X_1\)};
            \node[draw, circle, minimum size=0.85cm, right=of X_1] (D) {\(D\)};
            \node[draw, circle, minimum size=0.85cm, above=of M] (d) {\(d\)};
            \draw (Z_2) -- (X_3);
            \draw (X_3) -- (Y);
            \draw (Z_2) -- (X_2);
            \draw (X_2) -- (Y);
            \draw (X_2) -- (D);
            \draw (Z_1) -- (X_2);
            \draw (Z_1) -- (X_1);
            \draw (X_1) -- (D);
            \draw (d) -- (M);
            \draw (M) -- (Y);
        \end{tikzpicture}
        \par\vspace{0.5em}
        (b) SWIG after \(\text{fix}(D = d)\)
    \end{minipage}
\end{figure}

In this example, the goal is to estimate the causal effect of \(D\) on \(Y\), i.e., the mapping \(d \mapsto Y(d)\). The set of variables
\[
S = \{ \{X_1, X_2\}, \; \{X_2, X_3\}, \; \{X_2, Z_2\}, \; \{X_2, Z_1\} \}
\]
are valid adjustment sets that, when conditioned upon, block all backdoor paths between \(D\) and \(Y(d)\) in the SWIG. Notably, conditioning on \(X_2\) alone is insufficient, as it opens a path where \(X_2\) acts as a collider; adding \(X_1\), \(X_3\), \(Z_1\), or \(Z_2\) blocks this path, ensuring the required conditional ignorability.

\bigskip

An important and perhaps underappreciated advantage of the counterfactual DAG (or SWIG) approach is its ability to clarify not only which adjustment sets are valid, but also which are actually necessary and helpful for identifying causal effects. This is particularly relevant when considering the efficiency of estimators and the potential pitfalls of including superfluous control variables. Consider the simple causal DAG:
\[
Z \leftarrow D \rightarrow Y,
\]
where \(D\) is the treatment, \(Y\) is the outcome, and \(Z\) is a variable influenced by \(D\). In this structure, there are no backdoor paths from \(D\) to \(Y\) that could introduce confounding. The only paths from \(D\) to \(Y\) are direct, and all paths from \(D\) to \(Z\) are also direct. When we construct the corresponding counterfactual DAG or SWIG, representing the intervention \(\text{fix}(D = d)\), the graph becomes:
\[
Z(d) \leftarrow d \rightarrow Y(d).
\]
Here, both \(Z(d)\) and \(Y(d)\) are potential outcomes under the intervention \(D = d\), and \(d\) is a fixed value. In this counterfactual graph, there are no open paths from \(d\) to either \(Z(d)\) or \(Y(d)\) except the direct arrows, and crucially, there is no path from \(Z(d)\) to \(Y(d)\) that could induce spurious association. The key insight is that, under this counterfactual representation, no adjustment is required to identify the causal effect of \(D\) on \(Y\). Formally, the empty set is a valid adjustment set:
\[
Y(d) \perp\!\!\!\perp D.
\]
This means that the average causal effect of \(D\) on \(Y\) can be identified without controlling for any other variables—adjustment is unnecessary because there is no confounding to remove. However, it is also true that \(Z\) is a valid control variable in the sense that adjusting for \(Z\) does not introduce bias. This can be seen by considering a "cross-world" DAG (below) that combines both factual and counterfactual variables from the respective structural equation models.

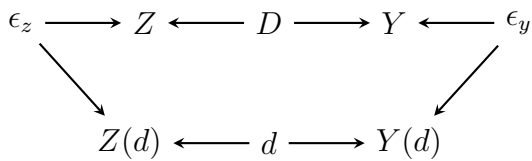
\[
    \begin{tikzpicture}[->, >=stealth, thick, node distance=1.0cm]
        \node[](ez) {\(\epsilon_z\)};
        \node[right=of ez] (z) {\(Z\)};
        \node[right=of z] (D) {\(D\)};
        \node[right=of D] (y) {\(Y\)};
        \node[right=of y] (ey) {\(\epsilon_y\)};
        \node[below=of D] (d) {\(d\)};
        \node[left=of d] (zd) {\(Z(d)\)};
        \node[right=of d] (yd) {\(Y(d)\)};
        \draw (D) -- (z);
        \draw (D) -- (y);
        \draw (ez) -- (z);
        \draw (ez) -- (zd);
        \draw (ey) -- (y);
        \draw (ey) -- (yd);
        \draw (d) -- (zd);
        \draw (d) -- (yd);
    \end{tikzpicture}
    \] In such a graph, \(Y(d)\) is d-separated from \(D\) by \(Z\), so the conditional independence
\[
Y(d) \perp\!\!\!\perp D \mid Z
\]
holds as well. Nevertheless, while \(Z\) is a valid control in the sense that it blocks no necessary paths and does not introduce bias, it is also superfluous: adjusting for \(Z\) does not improve identification and can in fact reduce the precision of our estimates. Including unnecessary variables in the adjustment set can lead to less efficient estimators, as it increases variance without reducing bias. This underscores the practical value of the counterfactual DAG approach: it not only identifies all valid adjustment sets but also helps to avoid overadjustment by highlighting when adjustment is unnecessary.

\section{The Backdoor Criterion}

A central question in causal inference is: under what conditions can we identify the causal effect of a treatment variable on an outcome using observational data? The backdoor criterion, formulated by Judea Pearl, provides a clear graphical answer to this question. This section presents the theorem, illustrates it with figures, and walks through a canonical example.

\bigskip

As discussed before, a \textbf{backdoor path} is any path from the treatment \(D\) to the outcome \(Y\) that starts with an arrow into \(D\), indicating a potential confounding influence. The \textit{backdoor criterion} provides a graphical condition for identifying a set of variables (an adjustment set) that, when conditioned upon, ensures the identification of the causal effect of \(D\) on \(Y\).

\bigskip

\begin{theorem}
\textbf{(Backdoor Criterion)} Let \(G\) be a DAG representing an acyclic structural equation model (ASEM). Relabel a treatment node \(X_j\) as \(D\), and let \(Y\) be any descendant of \(D\). A set of variables \(S\) is a \textbf{valid adjustment set}—meaning it implies conditional ignorability,
\[
Y(d) \perp\!\!\!\perp D \mid S,
\]
if the following two conditions hold:
\begin{enumerate}
    \item No element of \(S\) is a descendant of \(D\).
    \item All backdoor paths from \(D\) to \(Y\) are blocked by \(S\).
\end{enumerate}
\end{theorem} This criterion ensures that, by conditioning on \(S\), all spurious (non-causal) associations between \(D\) and \(Y\) are removed, leaving only the causal effect to be estimated. To visualize the backdoor criterion, consider the following DAG, adapted from Pearl’s classic example:

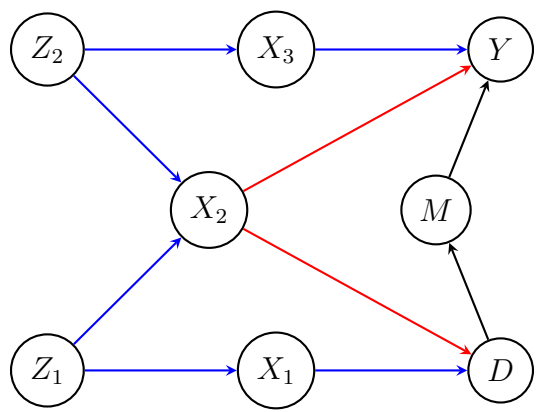
\begin{figure}[htbp]
    \centering
    \begin{tikzpicture}[->, >=stealth, thick, node distance=2cm]
        \node[draw, circle, minimum size=0.85cm] (X_2) {\(X_2\)};
        \node[draw, circle, minimum size=0.85cm, right=of X_2] (M) {\(M\)};
        \node[draw, circle, minimum size=0.85cm, above left=of X_2] (Z_2) {\(Z_2\)};
        \node[draw, circle, minimum size=0.85cm, right=of Z_2] (X_3) {\(X_3\)};
        \node[draw, circle, minimum size=0.85cm, right=of X_3] (Y) {\(Y\)};
        \node[draw, circle, minimum size=0.85cm, below left=of X_2] (Z_1) {\(Z_1\)};
        \node[draw, circle, minimum size=0.85cm, right=of Z_1] (X_1) {\(X_1\)};
        \node[draw, circle, minimum size=0.85cm, right=of X_1] (D) {\(D\)};
        \draw[blue] (Z_2) -- (X_3);
        \draw[blue] (X_3) -- (Y);
        \draw[blue] (Z_2) -- (X_2);
        \draw[red] (X_2) -- (Y);
        \draw[red] (X_2) -- (D);
        \draw (D) -- (M);
        \draw (M) -- (Y);
        \draw[blue] (Z_1) -- (X_2);
        \draw[blue] (Z_1) -- (X_1);
        \draw[blue] (X_1) -- (D);
    \end{tikzpicture}
    \caption{DAG illustrating backdoor paths from \(D\) to \(Y\). Red: direct/inner backdoor path. Blue: longer/outer backdoor path.}
\end{figure} 
\[
\text{(i) }D \leftarrow X_2 \rightarrow Y
\]
\[
\text{(ii) }{D \leftarrow X_1 \leftarrow Z_1 \rightarrow X_2 \leftarrow Z_2 \rightarrow X_3 \rightarrow Y}
\] 

\bigskip

In this DAG, there are two notable backdoor paths from \(D\) to \(Y\) as shown above. To apply the backdoor criterion, we must block all such paths by conditioning on an appropriate set \(S\). Conditioning on \(X_2\) alone blocks the inner backdoor path (i), since \(X_2\) is a non-collider on this path and conditioning on it blocks the flow of association. However, conditioning on \(X_2\) alone actually opens the outer backdoor path (ii), because \(X_2\) acts as a collider along that path. In general, conditioning on a collider (or its descendant) opens a path that would otherwise be blocked, potentially introducing bias. Therefore, to block the outer backdoor path as well, we must also condition on an additional variable that lies on that path but is not a descendant of \(D\)—for example, \(X_1\), \(X_3\), \(Z_1\), or \(Z_2\). Thus, valid adjustment sets include:
\[
S_1 = \{X_1, X_2\} \quad \text{or} \quad S_2 = \{X_2, X_3\} \quad \text{or} \quad S_3 = \{X_2, Z_1\} \quad \text{or} \quad S_4 = \{X_2, Z_2\}
\]
Each of these sets blocks all backdoor paths from \(D\) to \(Y\) and contains no descendants of \(D\).

\bigskip

It is important to note that conditioning on a descendant of \(D\), such as \(M\) (an intermediate outcome on the causal path from \(D\) to \(Y\)), is not valid. Doing so can introduce bias by blocking part of the direct causal effect or opening new spurious paths. Additionally, The backdoor criterion systematically yields minimal adjustment sets for identification. However, it may not capture every valid set. For example, consider the simple DAG:
    \[
    Z \leftarrow D \rightarrow Y.
    \]
Here, conditioning on \(Z\) does not satisfy the backdoor criterion (since \(Z\) is a descendant of \(D\)), yet \(Z\) is a valid control. In this case, \(D\) directly causes \(Y\) without confounding, so there is no need to adjust for \(Z\). Adjusting for \(Z\) may even lower the precision of the estimated effect. This limitation is useful as it helps disregard controls that, while valid, do not add any meaningful information for identifying the causal effect. The same observation applies within the counterfactual approach.
\section{Notebook}
The Jupyter notebook is available here in \href{https://github.com/Gauranga2022/CMI-MSc-Data-Science/tree/main/Sem4/CausalInference%20(GP)}{Github.}
\section{References}
\begin{enumerate}
  \item Victor Chernozhukov et al.\ (2023).  
        \emph{Applied Causal Inference Powered by Machine Learning and AI},
        Chapter 2.  
        \url{https://causalml-book.org/}
\end{enumerate}

\end{document}